\documentclass[12pt]{article}
\usepackage{fullpage,amsmath,amssymb,amsfonts,amsthm}
\usepackage{tikz}
\usepackage{url}
\usetikzlibrary{shapes}
\usepackage{a4wide}
\usepackage{caption}
\usepackage{titling}
\thanksmarkseries{arabic}
\usepackage{cite}
\usepackage{algorithm, algorithmic}
\usepackage{multicol}
\usepackage{hyperref}

\usepackage{caption}
\newtheorem{theo}{Theorem}[section]
\newtheorem{defi}[theo]{Definition}
\newtheorem{prop}[theo]{Proposition}

\newtheorem{lemma}[theo]{Lemma}
\newtheorem{conj}{Conjecture}
\newtheorem{post}{Postulate}

\DeclareMathOperator{\MWE}{MWE}

\DeclareMathOperator{\stor}{stor}

\newcommand{\Z}{\mathbb{Z}}
\newcommand{\N}{\mathbb{N}}

\newcommand{\Esp}{{\mathbb E}}
\newcommand{\Prob}{{\mathbb P}}

\usepackage{authblk}
\author[1]{Nicolas Delfosse}
\author[2]{Ben W. Reichardt}
\author[1]{Krysta M. Svore}
\affil[1]{\small Microsoft Quantum and Microsoft Research, Redmond, WA, USA}
\affil[2]{\small University of Southern California, Los Angeles, CA, USA}

\title{Beyond single-shot fault-tolerant quantum error correction}

\begin{document}

\maketitle

\begin{abstract}

Extensive quantum error correction is necessary in order to perform a useful computation on
a noisy quantum computer. 
Moreover, quantum error correction must be implemented based on imperfect 
parity check measurements that may return incorrect outcomes or inject 
additional faults into the qubits.
To achieve fault-tolerant error correction, Shor proposed to repeat the sequence
of parity check measurements until the same outcome is observed sufficiently 
many times. Then, one can use this information to perform error correction.
A basic implementation of this fault tolerance strategy requires $\Omega(r d^2)$ 
parity check measurements for a distance-$d$ code defined by $r$ parity checks.
For some specific highly structured quantum codes, Bombin has shown 
that single-shot fault-tolerant quantum error correction is possible using only $r$ measurements.
In this work, we demonstrate that fault-tolerant quantum error correction 
can be achieved using $O(d \log(d))$ measurements for any code
with distance $d \geq \Omega(n^\alpha)$ for some constant $\alpha > 0$.
Moreover, we prove the existence of a sub-single-shot 
fault-tolerant quantum error correction scheme using fewer than $r$ measurements.
In some cases, the number of parity check measurements required for fault-tolerant quantum
error correction is exponentially smaller than the number of parity checks 
defining the code.
\end{abstract}

As memory scales to higher density, error rates rise and new sources of error emerge, 
requiring extensive error correction.
Ultimately, at the quantum scale, any manipulation of a quantum system introduces
an error with a non-negligible probability.
In this work, we consider the problem of error correction with a faulty quantum device.
The presence of faults in parity check measurements significantly increases the cost of
quantum error correction in comparison with the perfect measurement case.
Our main goal is to reduce the time overhead of fault tolerance.

\medskip
Error correction with linear codes is based on the evaluation of parity checks 
which provide the {\em syndrome} that is then used to identify the possible error.
Faults may occur during the syndrome measurement resulting either in 
incorrect syndrome values or in additional errors injected in the data.
Linear codes can be used in the quantum setting thanks to the CSS construction
\cite{calderbank1996:css, steane1996:css} and the stabilizer 
formalism \cite{gottesman1997:stabilizer}.
In the present work, we do not consider the technical details of these 
constructions. We incorporate the quantum constraints as follows.

\medskip
\begin{post} [Quantum parity check constraint] \label{postulate}
If a linear code with parity check matrix $H$ is used for quantum error correction, 
the only measurements available are the parity checks that are linear combinations 
of the rows of $H$.
\end{post}

\medskip
This constraint is satisfied for all measurement schemes considered below.
Our motivation for focusing on this unique quantum constraint is two-fold. 
First, we want to emphasize the aspects of quantum fault tolerance 
that are of purely classical nature and that deserve classical solutions.
Second, we hope to make this work and the mathematical questions it raises 
accessible to a broader audience.

\medskip
Shor developed the first quantum error-correcting code in 1995 \cite{shor1995:qec}.
However, the presence of faults makes error correction 
challenging to implement in a quantum device since measurement outcomes 
cannot be trusted as Fig.~\ref{fig:repetition_code} shows.
The following year, Shor introduced a fault-tolerant mechanism in order to perform
quantum error correction with faulty components \cite{shor1996:ftqec}.
This line of work led to the threshold theorem \cite{aharonov2008:threshold} that demonstrates 
that an arbitrary long quantum computation can be performed over a fault quantum device at 
the price of a reasonable asymptotic overhead if the noise strength is below a certain threshold value.
An elegant proof of this result based on a notion of fault-tolerant computation for 
concatenated codes was proposed later \cite{aliferis2005:threshold, gottesman2010:threshold}.
Although asymptotically reasonable, the overhead required for fault tolerance is
daunting in the regime of practical applications \cite{fowler2012:surface, reiher2017:nitrogen}.

\begin{figure}
\centering
\includegraphics[scale=.7]{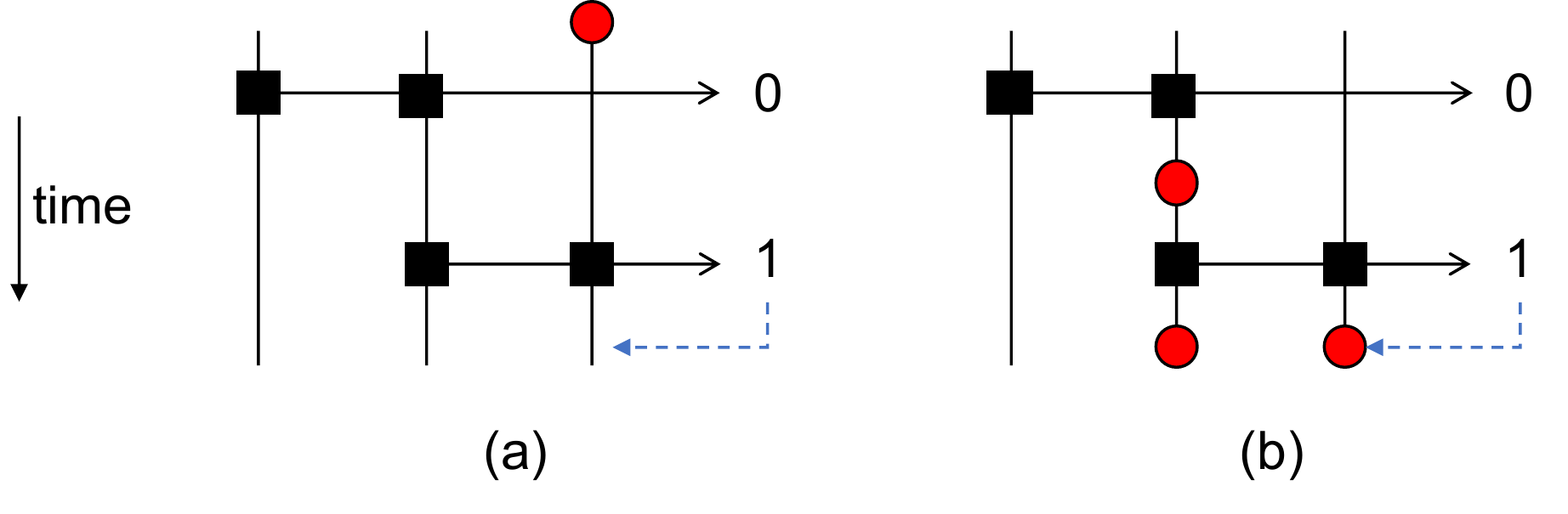}
\caption{
Error correction with the three-bit repetition code.
Errors affecting a codeword $(x_1, x_2, x_3)$ are corrected based on the
measurement of the two parity checks $x_1 + x_2 \mod 2$ and $x_2 + x_3 \mod 2$.
(a) An input bit flip on the third bit (red circle) is corrected based on the 
outcome $(0, 1)$ for the parity checks.
(b) The red circles indicates the presence of a fault after the first measurement that flips the
second bit. The outcome $(0, 1)$ obtained suggests an error on the third bit
resulting in two bit flips after correction.
This protocol must be made fault-tolerant in order to avoid the confusion of internal faults
with errors on the input codewords.
}
\label{fig:repetition_code}
\end{figure}

\medskip
It seems natural to consider longer measurement sequences in order to distinguish
between internal faults and input errors. However, this intuition is at tension 
with the fact that number of possible fault locations increases rapidly 
with the number of measurements.
A sequence of $m$ measurements for a code with length $n$ provides only $m$
outcome bits for about $mn$ fault locations. 
The number of possible fault combinations ($2^{n m}$) 
grows exponentially faster than the number of outcome vectors ($2^m$) 
when $n, m \rightarrow \infty$.
Luckily, a number of fault configurations can be considered equivalent given that they
have the same effect on codewords, making fault tolerance possible.

\medskip
The basic idea of Shor's fault-tolerant scheme is to repeat the syndrome measurement 
until we observe the same syndrome vector enough times on consecutive 
measurements.
We can then rely on the syndrome value and correct errors accordingly.
For a linear code encoding $k$ bits into $n$ bits with minimum distance $d$, 
the Shor scheme requires up to $((d+1)/2)^2$ repetitions of the syndrome measurement,
that is $(n-k)((d+1)/2)^2 \geq \Omega(d^3)$ parity check measurements.
This is because $(d+1)/2$ consecutive identical syndrome vectors are necessary 
to guarantee a correct syndrome value in the presence of up to $(d-1)/2$ faults.
This large time overhead is also present in other fault-tolerant quantum error
correction schemes.
Flag error correction \cite{chao2018:flagec, chamberland2018:flagec} 
considerably reduces the number of ancilla qubits but still leads to
a similar time overhead.
Steane method  \cite{steane1999:ftqec} implements the syndrome readout 
in constant depth, however the difficulty is transferred to the fault-tolerant 
preparation of an ancilla state.

\medskip
In the present work, Theorem~\ref{theo:time_overhead} proves 
that fault-tolerant error correction with an arbitrary linear code respecting 
the quantum constraint can be implemented with $O(d \log(d))$ parity check measurements
if the code distance grows polynomially with $n$, that is $d \geq \Omega(n^\alpha)$
for some $\alpha > 0$.
In the case of a family of codes with polylog distance $d \geq \Omega(\log(n)^\alpha)$,
a sequence of $O(d^{1 + \frac{1}{\alpha}})$ parity check measurements is enough
for fault-tolerant error correction.
This speed-up is doubly beneficial for error correction. On the one hand, it 
reduces the time per error correction cycle. On the other hand, fewer measurements 
means less noise during the correction cycle, improving the life-time of
encoded data.

\medskip
For a code defined by $r$ parity check equations, one may be tempted to conjecture
that at least $r$ measurements must be performed to achieve fault-tolerant error correction
since the code is not fully defined by fewer than $r$ parity checks.
Moreover, Bombin proved that {\em single-shot} fault-tolerant error correction, 
based on exactly $r$ parity check measurements,
is possible for a family of highly structured quantum codes \cite{bombin2015:single_shot}.
Surprisingly, our work demonstrates that one can go beyond single-shot and design 
fault-tolerant error correction schemes that use much fewer than $r$ measurements.
Applying Theorem~\ref{theo:time_overhead} for a family of codes with
positive encoding rate $k/n \rightarrow R$ with $0 < R < 1$ and minimum distance 
$d = \Omega(n^\alpha)$ for $0< \alpha < 1$, we obtain a fault-tolerant
error correction scheme based on $O(d \log(d)) = O(n^\alpha \log(n))$ parity check 
measurements for a code defined by $r = \Omega(n)$ parity checks.
The same result holds for polylog distance codes with positive rate.
In that case, the length of the measurement sequence $O(d^{1 + \frac{1}{\alpha}})$ 
is exponentially smaller than the number of checks $r$.
This result applies to standard families of codes with sub-linear 
distance such as 
turbo codes \cite{berrou1993turbo_codes},
cycle space of graph \cite{decreusefond1997error}, 
finite geometry codes \cite{kou2001finite_geo_LDPC},
and polar codes \cite{arikan2009polar_codes},
providing many examples of {\em sub-single-shot} fault-tolerant error correction schemes.

\medskip
We were not able to prove a lower bound that matches our upper bound
$O(d \log(d))$ on the length of fault-tolerant measurement sequences
for polynomial distance codes.
We conjecture that a linear length sequence can be achieved.

\medskip
\begin{conj} \label{conjecture}
For any family of $[n, k, d]$ linear codes with polynomial distance
$d \geq \Omega(n^\alpha)$ for some constant $\alpha > 0$,
there exists a sequence of $O(d)$ parity check measurements that allows
for fault-tolerant error correction.
\end{conj}

\medskip
The basic idea of our scheme is to extract as much safe information as possible
from the measurement of a redundant set of parity checks.
It was first noticed by \cite{fujiwara2014:syndrome_correction} that additional
measurements can be exploited to correct the syndrome values. 
Ashikhmin {\em et al.} \cite{ashikhmin2014:syndrome_correction, ashikhmin2016:syndrome_correction_LDGM} 
generalized this idea to arbitrary stabilizer codes.
This work also closely relates to the notion of single-shot error correction 
\cite{bombin2015:single_shot}
generalized recently by the work of Campbell \cite{campbell2019:single_shot}
which focuses on quantum Low Density Parity Check codes \cite{mackay2004:QLDPC, tillich2013:QLDPC}.
Our formalism applies to arbitrary linear codes and takes into account 
internal data errors which were not considered by
\cite{fujiwara2014:syndrome_correction, ashikhmin2014:syndrome_correction, 
ashikhmin2016:syndrome_correction_LDGM}.

\medskip
Section~\ref{sec:FTEC} reviews the concept of fault tolerance within
the formalism of linear codes. We propose a definition of fault tolerance
and we apply this definition to prove that fault-tolerant error correction
increases the lifetime of encoded data.
Section~\ref{sec:FTDECODING} focuses on the design of a notion 
of minimum distance and a minimum weight decoder adapted to the 
context of fault tolerance. 
Theorem~\ref{theo:time_overhead}, presented in Section~\ref{sec:BOUND},
is the main result of this paper.
Relying on the results of Section~\ref{sec:FTDECODING}, it provides an 
upper bound on the number of measurements necessary to make a 
linear code fault-tolerant.
Numerical results illustrating our scheme are presented in Section~\ref{sec:numerics}.
In particular, our Monte-Carlo simulations show that the lifetime of
data encoded with a fault-tolerant error correction scheme can surpass 
the lifetime of physical data.

\section{Fault-tolerant error correction} \label{sec:FTEC}

\subsection{Background on error correction} \label{subsec:FTEC:error_correction}

We consider error correction based on classical binary linear codes.
For a more complete treatment of this topic we refer 
to \cite{macwilliams1977:ECC, vanlint2012:coding_theory}.
A {\em linear code} $C$, or simply a code, 
with length $n$ and minimum distance $d$
is defined to be a $k$-dimensional subspace of $\Z_2^n$ such that 
the minimum Hamming weight of a non-zero codeword $x \in C$
is $d$. We denote by $[n, k, d]$ the parameters of the code or $[n, k]$
when the minimum distance is unknown.
If some bits of a codeword $x$ are flipped, it is mapped onto $y = x+e$ 
for some $e \in \Z_2^n$.
If the number of bit flip that occur satisfies $|e| \leq (d-1)/2$.
we can recover $x$ by selecting the closest codeword from $y$.

\medskip
A linear code can be defined by a generator matrix $G \in M_{k, n}(\Z_2)$,
such that the rows of $G$ form a basis of $C$. The code $C$ is the set of
vectors $x G$, where $x \in \Z_2^k$ and 
the transformation $x \mapsto x G$ is an encoding map.
Alternatively, a linear code can be given by a parity check matrix 
$H \in M_{r, n}(\Z_2)$, such that the codewords of $C$ 
are the vectors $x$ with $x H^T = 0$.

\medskip
For example, the Hamming code with parameters $[6, 4, 3]$ is defined by 
the parity check matrix
\begin{align} \label{eqn:parity_check_matrix_hamming}
H = 
\begin{pmatrix}
1 & 0 & 1 & 0 & 1 & 0 & 1 \\ 
0 & 1 & 1 & 0 & 0 & 1 & 1 \\ 
0 & 0 & 0 & 1 & 1 & 1 & 1
\end{pmatrix} \cdot
\end{align}
The two following generator matrices
\begin{align} \label{eqn:generator_matrices}
G_1 = 
\begin{pmatrix}
1 & 0 & 0 & 1 & 1 & 0 \\
0 & 1 & 0 & 1 & 0 & 1 \\
0 & 0 & 1 & 0 & 1 & 1
\end{pmatrix}
\quad 
\text{ and }
\quad 
G_2 = 
\begin{pmatrix}
1 & 0 & 0 & 1 & 0 & 1 & 1 & 0 & 0 & 1 \\
0 & 1 & 0 & 1 & 1 & 0 & 1 & 0 & 1 & 0 \\
0 & 0 & 1 & 1 & 1 & 1 & 1 & 1 & 0 & 0
\end{pmatrix}
\end{align}
define two linear codes with parameters $[6, 3, 3]$ and $[10, 3, 5]$ respectively.

\medskip
Assume that an error $e$ occurs on a codeword $x$ in the code $C$, resulting 
in $x' = x+e$.
Error correction is based on the computation of the {\em syndrome} $s = (x+e) H^T = e H^T$.
A non-trivial syndrome indicates the presence of an error.
The value of the syndrome depends only on the error $e$.
By {\em decoding} we mean estimating the error $e$ given its syndrome $s$.
A decoder is a map 
$
D: \Z_2^{r} \longrightarrow \Z_2^{n}.
$
The decoding is said to be {\em successful} when $e$ is correctly identified by the decoder
that is if $D(s) = e$.
For practical purposes, an efficient implementation of the map $D$ is required.

\medskip
We call {\em minimum weight error (MWE) decoder}, a decoder that 
returns an error with minimum weight among the errors with syndrome $s$,
where $s$ is the observed syndrome.
In what follows, $D_{\MWE}^H$ denotes a MWE decoding map for the code $C$
with parity check matrix $H$.
A MWE decoder successfully identifies any error $e$ with weight up to $(d-1)/2$.

\medskip
A standard noise model in information theory is the {\em binary symmetric channel}
with crossover probability $p$. Each bit is flipped independently with probability $p$.
A MWE decoder can be used to correct this type of noise since when $p<1/2$
the error $\hat e = D_{\MWE}^H(s)$ returned for a syndrome $s$ is a most likely error (MLE)
for the binary symmetric channel, 
{\em i.e.} it maximizes the conditional probability $\Prob(e | s)$ among
the errors with syndrome $s$. We use the notation $\hat e$ to refer to
an estimation of the error $e$.

%

\subsection{Circuit error} \label{subsec:FTEC:circuit_error}

We are interested in the design of fault-tolerant error correction schemes
that work even when measured syndromes are noisy and data errors can
be introduced during the correction steps.

\medskip
Our goal is to protect a set of $n_D$ {\em data bits} using a fixed linear code $C_D$ 
with parameters $[n_D, k_D, d_D]$. Figure~\ref{fig:tanner0} presents our notations
for the error model.
We refer to the error present on the data bits before correction as the 
{\em input error} denoted $e^0$.
In order to correct errors with $C_D$, a sequence of $n_M$ measurements
is applied to the data bits;
each measurement returns the parity $m_i \in \{0, 1\}$ of a subset of 
the data bits. 
In the fault-tolerant setting, the bit $m_i$ may be flipped, resulting in
the outcome vector $m = m(e^0) + f \in Z_2^{n_M}$ where
$m(e^0) = e^0 H_D^T$ is the ideal measurement outcome and 
$f$ is called {\em measurement error}.
Faults in the measurement device also affect the data bits.
We model this source of error as a bit flip occurring after each
parity check measurement. 
Denote by $e^i \in \Z_2^{n_D}$ the {\em level-$i$ internal error} that occurs 
after the $i$ th measurement for $i=1, 2, \dots, n_M$.

\begin{figure}
\centering
\includegraphics[scale=.5]{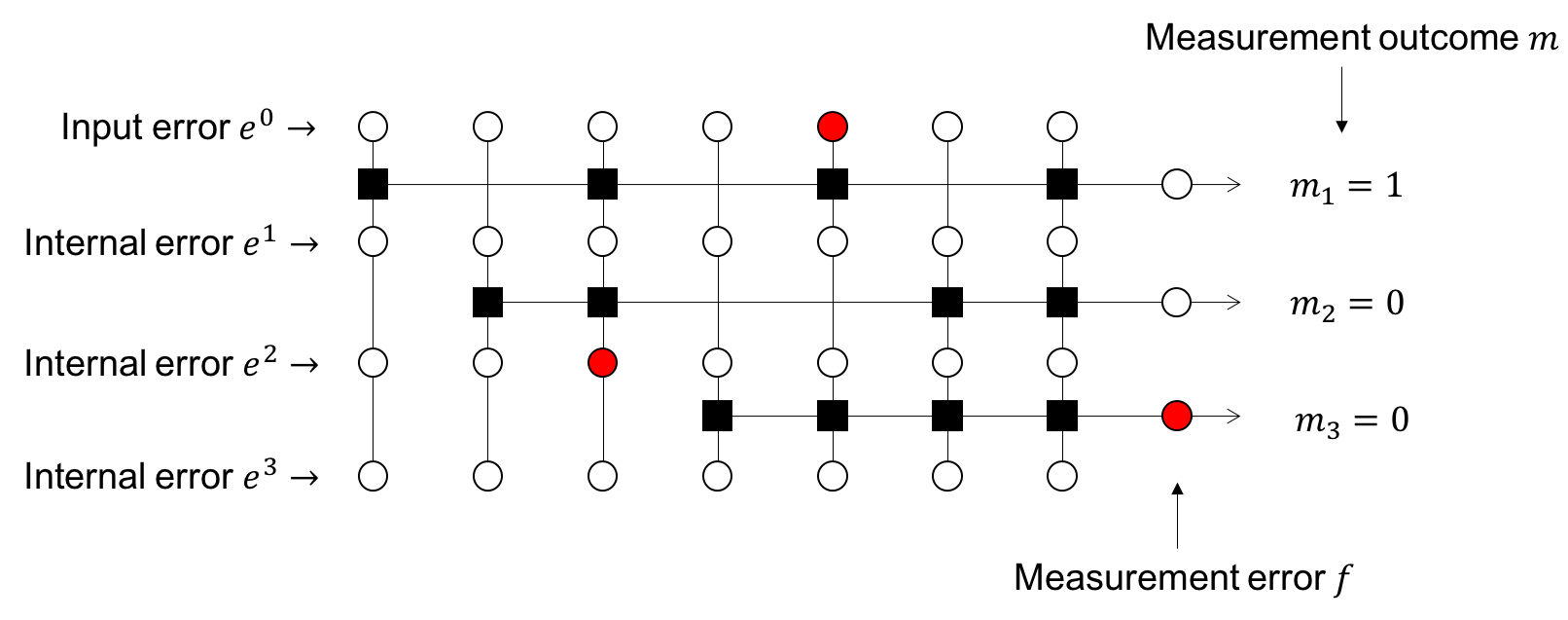}
\caption{Errors during the syndrome measurement for the Hamming code. 
A sequence of three measurements given by the three rows of the 
parity check matrix~\eqref{eqn:parity_check_matrix_hamming} is realized.
Circles indicates bit flip locations.  
A row of squares connected horizontally corresponds to a parity check measurement
between the bits marked by a square.
Red circles show a circuit error with one input error, one internal error and one measurement
error.
The input error is $e^0 = (0000100)$. 
An internal error occurs on the third data bit after measurement of $m_2$,
that is $e^2 = (0010000)$. The third outcome is flipped which means $f = (001)$.
The observed syndrome is $m = (1, 0, 0)$.}
\label{fig:tanner0}
\end{figure}

\medskip
Overall a {\em circuit error} is a pair $\varepsilon = (e, f)$ with
\begin{itemize}
\item Input error: $e^0 \in \Z_2^{n_D}$ on data bits.
\item Internal error: $(e^1, e^2,  \dots,  e^{n_M}) \in (\Z_2^{n_D} )^{n_M}$ on data bits.
\item Measurement error: $f \in \Z_2^{n_M}$ on measurement outcomes.
\end{itemize}
A circuit error is a binary vector of length $(n_D+1)(n_M+1)-1$.
The Hamming weight of a circuit error $\varepsilon = (e, f)$ is denoted by 
$|\varepsilon| = |e^0| + \dots + |e^r| + |f|$.

\subsection{Fault-tolerant decoder}  \label{subsec:FTEC:ft_decoder}

Error correction aims at identifying the data error, but this is a moving target
since internal errors can occur during the correction.

\medskip
Our goal is to correct the effect of a circuit error $\varepsilon$ 
on the data given the measurement outcome $m(\varepsilon)$.
One could aim at identifying the exact circuit error $\varepsilon$, 
but this is too ambitious because many circuit errors lead to 
the same outcome $m$.
Given $m = m(\varepsilon)$, our objective is to determine  
{\em residual data error} defined by
$$
\pi(\varepsilon) = \sum_{i=0}^{n_M} e^i
$$
after a sequence of $n_M$ measurements.

\medskip
A {\em decoder} is a map 
$D: \Z_2^{n_M} \rightarrow \Z_2^{n_D}$ 
that estimates the residual data error given the outcome $m$ observed.
When no confusion is possible, we denote by $\pi$ the residual error 
$\pi(\varepsilon)$ and the estimation returned by the decoder 
is denoted by $\hat \pi = D(m(\varepsilon))$.

\medskip
Aiming at identifying the exact residual error $\pi$ is still too 
ambitious. Some internal bit flips occur too late to be recognized. 
By trying to correct those late errors, we might actually inject 
additional errors. The following lemma makes this idea rigorous.
The {\em level} of a circuit error $\varepsilon$ is the first level 
$j$ such that $e^j \neq 0$.

\begin{lemma} \label{lemma:late_errors_amplification}
For any decoder $D$ we have
\begin{itemize}
\item Either $D$ corrects no circuit error $\varepsilon$ of level $n_M-1$,
{\em i.e.} $D(m(\varepsilon)) = 0$,
\item Or $D$ amplifies at least one error $\varepsilon$, 
{\em i.e.} $|\pi(\varepsilon) + D(m(\varepsilon))| > |\pi(\varepsilon)|$.
\end{itemize}
\end{lemma}

\begin{proof}
Assume that the last measurement involves $s \geq 2$ data bits.
A level-$(n_M - 1)$ error either results in a trivial outcome or yields
$m = (0 \dots 0 1)$.
Since $s \geq 2$, at least two distinct level-$(n_M - 1)$ errors 
$\varepsilon$ and $\varepsilon'$ lead to the outcome $m$. 
If $D$ corrects one of them, say $D(m) = \pi(\varepsilon)$, 
then the error $\varepsilon'$ is amplified.
\end{proof}

\medskip
The previous lemma presented motivates the following definition
of fault tolerance.
\begin{defi} \label{def:fault_tolerant_decoder}
A {\em fault-tolerant decoder} is defined to be a map 
$D: \Z_2^{n_M} \rightarrow \Z_2^{n_D}$ 
such that for all circuit error $\varepsilon = (e, f)$ such that 
$|\varepsilon| \leq (d_D-1)/2$ we have 
\begin{align} \label{eq:ft_condition}
|\pi + \hat \pi| \leq |f| +  \sum_{i=1}^{n_m} |e^{i}| = |\varepsilon| - |e^0|
\end{align}
where $\pi = \pi(\varepsilon)$ is the residual data error and 
$\hat \pi$ is the estimation of $\pi$ 
returned by the decoder.
\end{defi}

Roughly speaking, a fault-tolerant decoder corrects the input error $e^0$ 
without amplifying any internal error or measurement error.
Gottesman considers a notion of fault-tolerant quantum error correction
based on two conditions (ECA and ECB in \cite{gottesman2010:threshold}).
Our definition is similar to ECB. We do not need ECA that is useful 
for code concatenation in \cite{gottesman2010:threshold}.
Proposition~\ref{prop:ave_life_time_increase} below proves that 
Definition~\ref{def:fault_tolerant_decoder} provides a satisfying
notion of fault-tolerant error correction.

\subsection{Storage lifetime} \label{subsec:FTEC:life_time}

In this section, we will prove that encoded data constantly corrected with a fault-tolerant decoder 
can be preserved for a longer time than raw data.
Our proof of this property can be seen as a basic application of the 
rectangle method \cite{aliferis2005:threshold}.
This provides another justification for the definition of fault tolerance proposed
in Section~\ref{subsec:FTEC:ft_decoder}.
Figure~\ref{fig:storage_model} summarizes the basic idea of the storage noise model 
and the rectangle method is illustrated with 
Figure~\ref{fig:storage_model_rectangles}.

\medskip
Consider some data stored in an imperfect device and assume that,
at each time step, stored bits are flipped independently with probability 
$p$.
On any given bit, an error occurs in average after $1/p$ time steps. 

\begin{figure}
\centering
\includegraphics[scale=.4]{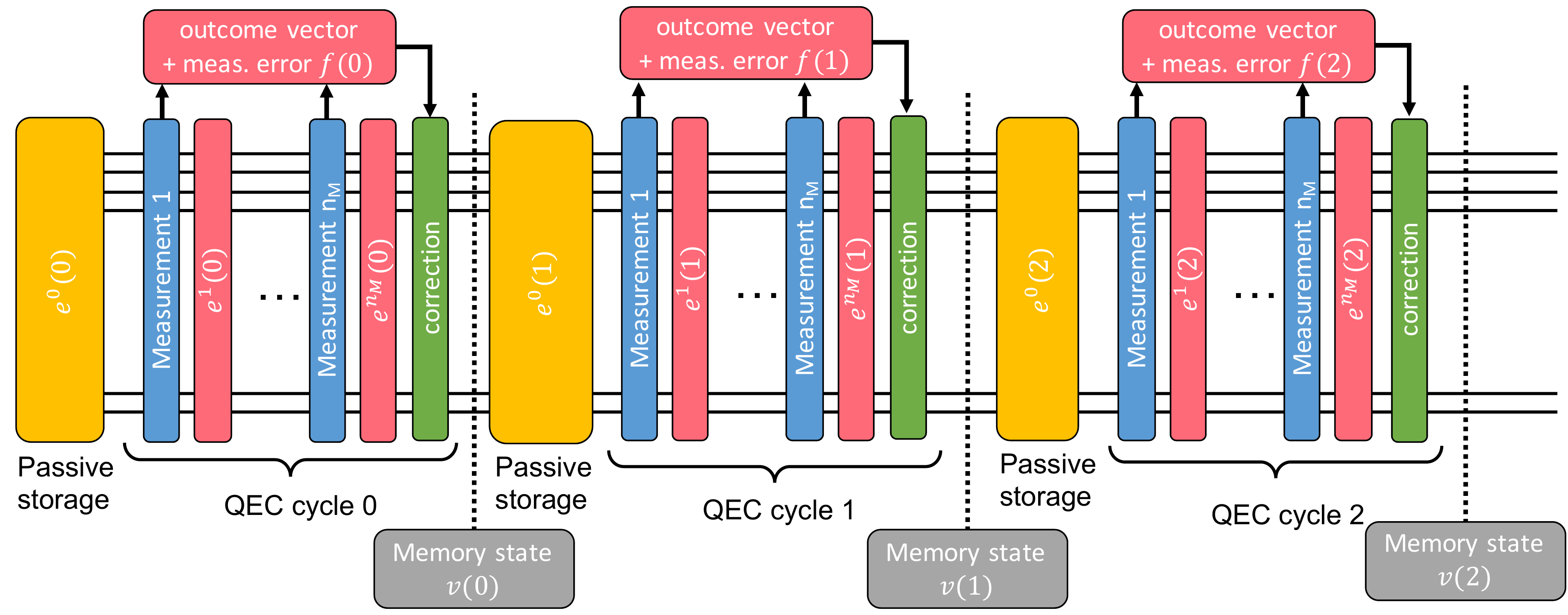}
\caption{Storage errors and memory state over three correction cycles.
The horizontal wires represent the $n_D$ data bits and 
time flows from left to right.
We alternate between rounds of passive storage and error-correction cycles.
The size of the blocks do not represent their duration.
}
\label{fig:storage_model}
\end{figure}

\medskip
In order to extend the lifetime of our data, we store encoded data
using a code $C_D$ with parameters $[n_D, k_D, d_D]$.
The state of the $n_D$ stored bits is described by a vector $v \in \Z_2^{n_D}$
called {\em memory state}.
If $v = c$ is a codeword of $C_D$, it represents some encoded information.
The same information $c$ can still be recovered from the vector $v = c+e^0$,
affected by a low-weight error $e^0$.
We say that a memory state 
$v \in \Z_2^{n_D}$ 
{\em stores the information $c \in C_D$}
if $c$ is the unique closest codeword of $C_D$ from the vector $v$.
We consider that a state $v$ that admits multiple closest codewords 
does not store any information.
For any error $e^0$ such that $|e^0| \leq (d_D-1)/2$, 
the information stored in a vector $v = c+e^0$ can be extracted by running 
a MWE decoder for the code $C_D$.

\medskip
In order to protect the stored data against the accumulation of errors,
we regularly run a fault-tolerant decoder. 
We alternate between passive storage and rounds of error correction.
Denote by $e^0(t)$ the $n_D$-bit error that accumulates on the memory state 
during the $t$ th storage round for $t \in \N$.
Let $f(t)$ and $e^1(t), \dots, e^{n_M}(t)$ be the measurement error and 
the internal errors that appear during the $t$ th correction round 
which takes $e^0(t)$ as an input error.
As one can see in Figure~\ref{fig:storage_model},
this defines a sequence of circuit errors 
$\varepsilon(t) = ( e(t), f(t) )$ for each time step $t \in \N$ that we call {\em storage error}.

\medskip
Consider the sequence $v(t)_{t \in \N}$ of memory states obtained after each 
round of error correction.
The {\em storage lifetime} of $c$ in the sequence $v(t)_{t \in \N}$ is defined to be
the first time step $t$ such that $v(t)$ does not store $c$ anymore.
The storage lifetime depends only on the storage error 
$\varepsilon(t)_{t \in \N}$ 
and not on $c$.
In this work, $\ell(\varepsilon)$ denotes the storage lifetime for a storage error 
$\varepsilon = \varepsilon(t)_{t \in \N}$.
The following lemma provides a sufficient condition to ensure that 
the stored data is not lost. It can be seen as a simple case of the rectangle 
method introduced in \cite{aliferis2005:threshold} as one can see in 
Fig~\ref{fig:storage_model_rectangles}.

\begin{lemma}[rectangle method] \label{lemma:FTEC:life_lower_bound}
Consider storage device equipped with a fault-tolerant decoder.
Let $\varepsilon(t)_{t \in \N}$ be a storage error that satisfies 
\begin{align} \label{eqn:lemma:life_lower_bound}
|\varepsilon(t-1)| - |e^0(t-1)| + |\varepsilon(t)| \leq  (d_D-1)/2
\end{align} 
for all $t = 0, \dots, N$, using the convention $|\varepsilon(-1)| = |e^0(-1)| = 0$.
Then, the storage lifetime $\ell(\varepsilon)$ is at least $N$.
\end{lemma}

\begin{proof}
We use the notation $v(t)$ for the memory state after the correction round $t$.
Without loss of generality we can assume that the initial memory state is $c = 0$. 
In order to prove that the information stored is preserved throughout the $N$ first 
rounds of correction, it suffices to show that the error $\varepsilon(t) = (e(t), f(t))$ 
treated by correction round $t$ satisfies 
\begin{align} \label{eqn:lemma:life_lower_bound:proof}
|\varepsilon(t)| \leq (d_D - 1)/2,
\end{align}
for all $t = 0, \dots, N$.
Indeed, by definition of fault tolerance, this implies that the memory state $v(t)$  
after correction has weight at most $(d_D - 1)/2$ proving that the information stored 
is not lost.

We can prove by induction that $\varepsilon(t)$ satisfies 
Eq.~\eqref{eqn:lemma:life_lower_bound:proof}.
The input error for the first round of correction is $e^0(0)$ that 
satisfies Eq.~\eqref{eqn:lemma:life_lower_bound:proof} by 
assumption~\eqref{eqn:lemma:life_lower_bound}.
Assume now that $\varepsilon(t-1)$ satisfies the 
inequality~\eqref{eqn:lemma:life_lower_bound:proof} for some $1 \leq t \leq N$.
Then, after correction it remains an error $v(t-1)$ such that 
$|v(t-1)| \leq |f(t-1)| + \sum_{i=1}^{n_M} |e^{i}(t-1)|$.
The input error of the next correction round (round $t$) is then 
$v(t-1) + e^{0}(t)$. It satisfies
\begin{align*}
|v(t-1)| + |e^{0}(t)| + |f(t)| + \sum_{i=1}^{n_M} |e^{i}(t)|
	& \leq |f(t-1)| + \sum_{i=1}^{n_M} |e^{i}(t-1)| + |\varepsilon(t)| \\
	& = |\varepsilon(t-1)| - |e^0(t-1)| + |\varepsilon(t)| \\
	& \leq (d_D-1)/2 \cdot
\end{align*}
The last inequality is the application of the hypothesis~\eqref{eqn:lemma:life_lower_bound}.
This proves Eq.~\eqref{eqn:lemma:life_lower_bound:proof}, concluding the proof of the lemma.
\end{proof}

\begin{figure}
\centering
\includegraphics[scale=.4]{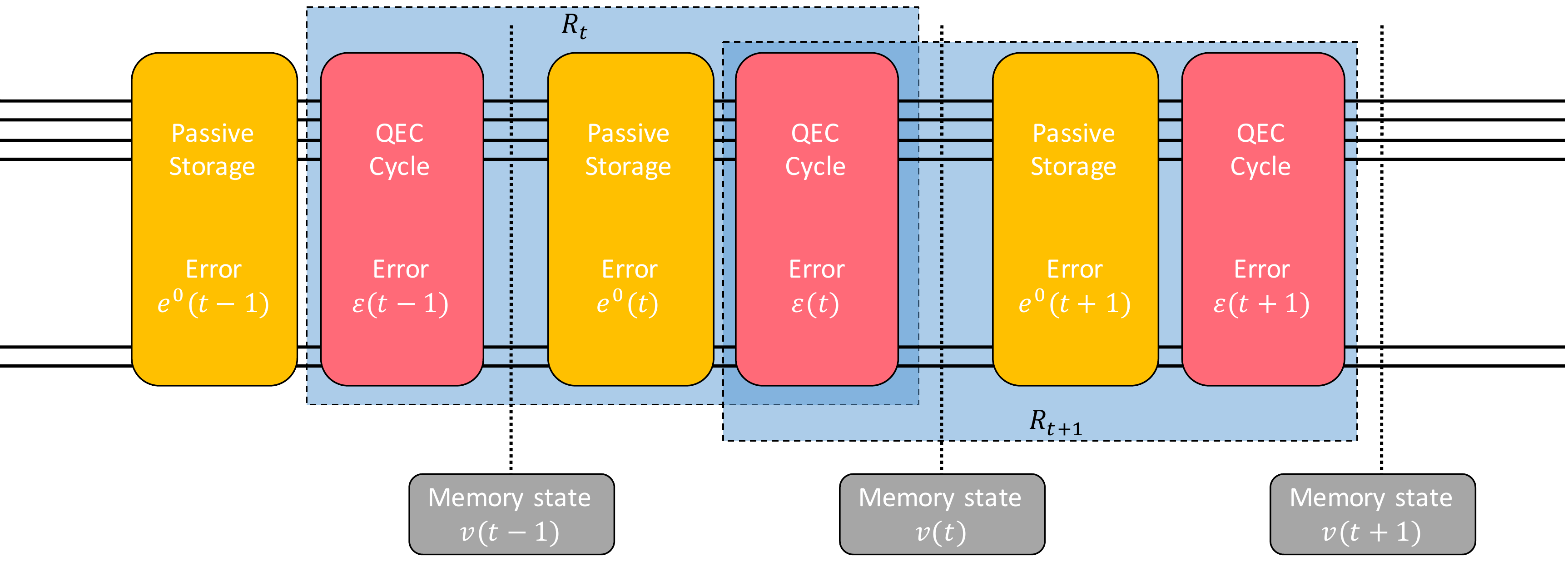}
\caption{
Illustration of the rectangle method.
Eq.~\eqref{eqn:lemma:life_lower_bound} means that the rectangle $R_{t}$
contains less than $(d_D-1)/2$ bit flips.
Lemma~\ref{lemma:FTEC:life_lower_bound} guarantees
that the memory state $v(t)$ stores the same data as the initial state $v(0)$
assuming that all the rectangles $R_0, \dots, R_{t}$ contains less than 
$(d_D-1)/2$ bit flips.
}
\label{fig:storage_model_rectangles}
\end{figure}

\medskip
Let us define a noise model for the storage.
During a storage round, the memory state bits are each flipped independently 
with probability $p$.
The outcomes measured during error correction rounds are flipped independently 
with probability $p$.
Moreover, during each parity-check measurement, each data bit is flipped 
independently with probability $p$.
In what follows, we denote by $\Prob_{\stor, p}$ the {\em storage distribution}
induced over the set of storage errors $\varepsilon(t)_{t \in \N}$.
To simplify, we assume the same probability for all bit flips.
This is a reasonable assumption to provide a proof that error correction 
increases the storage lifetime if the noise strength $p$ is small enough. 
For practical applications, we can adjust the different flip probabilities 
and introduce correlations that match the device's benchmarking results.

\begin{prop} \label{prop:ave_life_time_increase}
The probability that the storage lifetime is shorter than $N$ is upper 
bounded as follows
$$
\Prob_{\stor, p}( \ell(x) < N ) \leq N \binom {m}{ s } p^{ s }
$$
where $s = \frac{d_D + 1}{2}$ and $m = 2 n_M + 2 n_D n_M + n_D$.
\end{prop}

\begin{proof}
For $t \in \N$, let $A_t$ be the set of storage errors such that 
Eq.~\eqref{eqn:lemma:life_lower_bound} fails,
that is $|\varepsilon(t-1)| - |e^0(t-1)| + |\varepsilon(t)| > (d_D-1)/2$.
Denote by $s$ the integer $s = \frac{d_D +1}{2}$ 
The probability of the event $A_t$ is upper bounded as follows.
$$
\Prob_{\stor} ( A_t ) \leq \binom {m}{ s } p^{ s }
$$
where $m = 2 n_M + 2 n_D n_M + n_D$.

The lifetime is shorter than $N$ only if at least one of the events 
$A_t$ occurs with $t < N$. Therefore, a union bound
$$
\Prob_{\stor}( \ell(\varepsilon) < N ) \leq \sum_{t} \Prob_{\stor}( A_t )
$$
proves the proposition.
\end{proof}

The upper bound on the failure probability of Theorem~\ref{theo:time_overhead}
can be made arbitrarily small by selecting a code $C_D$
with large minimum distance under the condition that one can design a 
fault-tolerant decoder and that $m$ grows polynomialy with the minimum distance $d_D$.
This guarantees the exponential decays of the failure rate 
as 
$
O(p^{(d_D+1)/2}).
$
In the rest of this paper, we prove that one can design sequences of 
measurements that satisfy these conditions for a arbitrary codes.

\section{Fault-tolerant decoding} \label{sec:FTDECODING}

The main purpose of this section is to design a fault-tolerant analogue
of the minimum weight error decoder.

Consider first a naive generalization of the minimum weight error
decoder for linear codes. 
Given an outcome $m$, pick a minimum weight circuit error 
$\hat \varepsilon = \tilde D_{\MWE}(m)$ that reaches this outcome. 
The residual data error $\hat \pi = \pi(\hat x)$ could be used as a correction.
Unfortunately, this decoder does not satisfy the fault tolerance definition. 
It attempts to correct some bit flips that occur too late
to be identified, as illustrated by Lemma~\ref{lemma:late_errors_amplification},
resulting amplified data errors. 
In order to make the minimum weight circuit error strategy viable, we 
will restrict the action of the decoder to bit flips that occurs at early 
stages of the measurement sequence.

In this section, we first review standard techniques allowing to correct the
measurement outcome. Then, we focus on the correction of internal
errors. We introduce a notion of distance for the fault-tolerant setting
and we design a fault-tolerant decoder.

\subsection{Correction of measurement error by syndrome encoding} 
\label{subsec:FTDECODING:syndrome_enc}

Assume for now that no internal error occurs, {\em i.e.}
$e^1 = \dots = e^{n_M} = 0$, and focus on correcting
the input error with faulty measurements.
The basic idea is to measure an encoded version of the syndrome
in order to be able to correct measurement error and to estimate
the input error.
This strategy was introduced by Fujiwara \cite{fujiwara2014:syndrome_correction} 
and Ashikhmin {\em et al.} \cite{ashikhmin2014:syndrome_correction, ashikhmin2016:syndrome_correction_LDGM}
in the context of general quantum error-correcting codes. 
It is also a common strategy when working with topological quantum codes
\cite{dennis2002:tqm, campbell2019:single_shot}.
In the classical setting, redundant parity-check matrices 
appear to improve the performance of the belief propagation decoder
as observed with Low Density Parity Check (LDPC) codes based on finite geometry~\cite{kou2001:finite_geometry_LDPC}.

\medskip
The data code $C_D$ is given by a $r_D \times n_D$ parity check matrix $H_D$
with $r_D = n_D - k_D$. 
In order to protect the syndrome 
$s = e^0 H_D^T \in \Z_2^{r_D}$ 
of the input error $e^0$, we encode $s$ using a {\em measurement code} $C_M$
with generator matrix $G_M$. 
The measurement code is a $[n_M, k_M, d_M]$ linear code with $k_M = r_D$.
The {\em encoded syndrome} $m(e^0)$ for an input error $e^0$ is then
$
m(e^0) = s G_M = e^0 H_D^T G_M.
$
Equivalently, this is the syndrome of $e^0$ associated with the redundant 
parity check matrix 
$
H_m = G_M^T H_D.
$
We refer to this matrix as the {\em measurement matrix}.
Notice that this scheme respects Postulate~\ref{postulate}.
If an input error $e^0$ and a measurement error $f$ 
occur, the measurement outcome 
$$
m(e^0, f) = m(e^0) + f =  e^0 H_m^T + f
$$
is obtained.

\medskip
Redundancy in the measurement matrix can be used to correct measurement outcomes.
A measurement code with distance $d_M$ corrects at least $(d_M-1)/2$ bits.
Hamming $[7, 4, 3]$-code can be used in combination with any measurement 
code with dimension $k_M = 3$ to encode the 3-bit syndrome vector. 
One can select the smallest code with $k_M = 3$ that achieves a 
distance $d_M = 3$ or 5 from Grassl's code table 
\cite{Grassl:codetable_article, Grassl06:codetable_algo, Magma, Brouwer98:codetable_v0}.
Generator matrices for these codes are provided in 
Eq.~\eqref{eqn:generator_matrices} and the corresponding 
measurement matrices are
\begin{align} \label{eqn:measurement_matrices_Hamming}
H_{m, 1} = 
\begin{pmatrix}
1 & 0 & 1 & 0 & 1 & 0 & 1 \\ 
0 & 1 & 1 & 0 & 0 & 1 & 1 \\ 
0 & 0 & 0 & 1 & 1 & 1 & 1 \\
1 & 1 & 0 & 0 & 1 & 1 & 0 \\ 
1 & 0 & 1 & 1 & 0 & 1 & 0 \\ 
0 & 1 & 1 & 1 & 1 & 0 & 0
\end{pmatrix}
\quad
\text{ and }
\quad
H_{m, 2} = 
\begin{pmatrix}
1 & 0 & 1 & 0 & 1 & 0 & 1 \\
0 & 1 & 1 & 0 & 0 & 1 & 1 \\
0 & 0 & 0 & 1 & 1 & 1 & 1 \\
1 & 1 & 0 & 1 & 0 & 0 & 1 \\
0 & 1 & 1 & 1 & 1 & 0 & 0 \\ 
1 & 0 & 1 & 1 & 0 & 1 & 0 \\ 
1 & 1 & 0 & 1 & 0 & 0 & 1 \\
0 & 0 & 0 & 1 & 1 & 1 & 1 \\
0 & 1 & 1 & 0 & 0 & 1 & 1 \\
1 & 0 & 1 & 0 & 1 & 0 & 1
\end{pmatrix}
\end{align}
Both matrices define a sequence of measurements for the Hamming code
allowing for the correction of one or two flipped outcomes.
Then the corrected syndrome can be used to correct the data bits.
A larger minimum distance~$d_M$ allows for correcting more measurement errors.

\medskip
In general, it is better to correct both input error and measurement 
error simultaneously instead of sequentially correcting syndrome values 
and then data bits. 
Given an outcome~$m$, one can identify a minimum weight pair~$(e^0, f)$ of 
input error and measurement error that produces the outcome~$m$. 
With this strategy only five measurements suffice (the first five row of $H_{m,1}$) 
to correct a single bit flip either on the input data or on the outcome with the Hamming code. 
The measurement code is a [5,3,2] linear code. 

\medskip
Fujiwara \cite{fujiwara2014:syndrome_correction} 
and Ashikhmin {et al.} \cite{ashikhmin2014:syndrome_correction, ashikhmin2016:syndrome_correction_LDGM} 
designed measurement matrices suited for stabilizer
codes and obtained bounds on the number of measurements required for correcting 
input and measurement errors.
However, in order to make error correction applicable to a realistic setting, we must
also include internal errors.
In the remainder of this paper, we design a fault-tolerant error correction scheme 
that tolerates internal errors at the price of a moderate increase of the number 
of measurements required.

\subsection{Sequential Tanner graph and cluster decomposition}  \label{subsec:FTDECODING:tanner}

The Tanner graph \cite{tanner1981:tanner_graph} is convenient tool 
for designing error-correcting codes and their decoders \cite{richardson2001:ldpc}.
In our context, we associate a {\em sequential Tanner graph} 
with a $n_M \times n_D$ measurement matrix $H_m$.
Figure~\ref{fig:tanner2} shows two representations of a 
circuit error using the sequential Tanner graph for the Hamming code
equipped with the measurement matrix $H_{m,1}$ given in 
Eq.~\eqref{eqn:measurement_matrices_Hamming}.

\medskip
The standard Tanner graph used in classical coding theory encodes the set of 
all the parity check measurements.
Our sequential Tanner graph contains additional information such as the order in 
which measurements are realized.
This information is necessary in order include outcome flips dues to internal errors. 
This Tanner graph can be seen as a sequential version of the Tanner graph used
for instance in the context of topological quantum codes or quantum LDPC codes 
\cite{dennis2002:tqm, kovalev2013:qldpc, gottesman2014:ldpc}
with additional nodes for measurement errors.

\medskip
The sequential Tanner graph generalizes the diagram of Figure~\ref{fig:tanner0}.
There are $n_M+1$ rows of nodes that correspond to $n_M+1$ levels of data errors 
$e^0, e^1, \dots, e^{n_M}$ from top to bottom.
Denote this set of nodes by 
$$
V_D = \{ v_{i, j} \ | \ (i, j) \in [0, n_M] \times [1, n_D] \} \cdot
$$
For $j=1, \dots, n_D$, each node in the sequence 
$v_{0, j}, v_{1, j}, \dots, v_{n_M, j}$ 
is connected to its successor. 
Two consecutive rows $e^{i-1}$ and $e^{i}$ are separated by a row of check nodes (square)
indicating the bits involved in the $i$ th parity check measurement $m_i$.
A node is added at the end of each check node row to mark the measurement outcome
flip. Let 
$$
V_M = \{ u_i \ | \ i \in  [1, n_M] \}
$$ 
be this set of nodes.

\begin{figure}
\centering
\includegraphics[scale=.5]{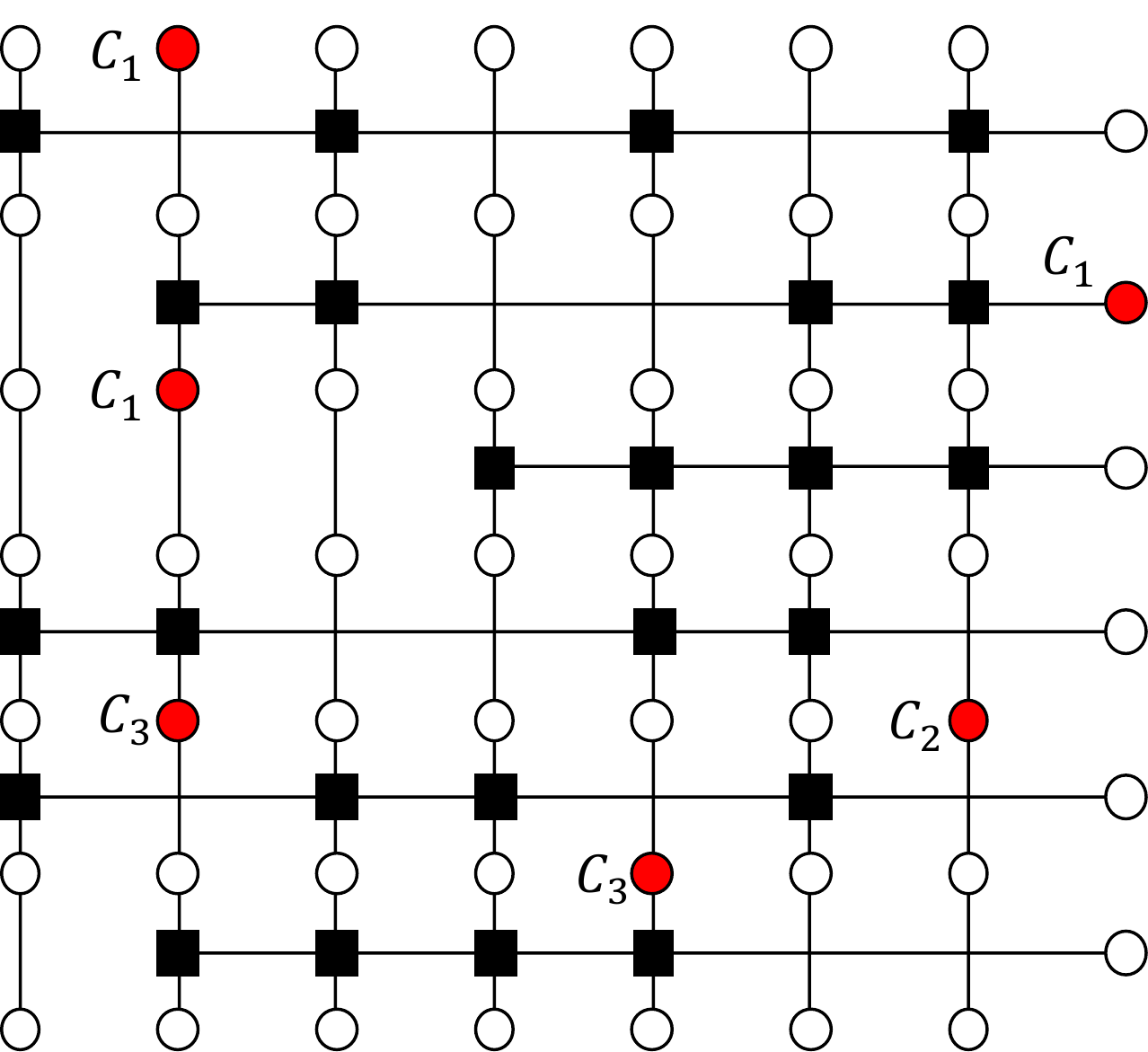}
\hspace{1cm}
\includegraphics[scale=.5]{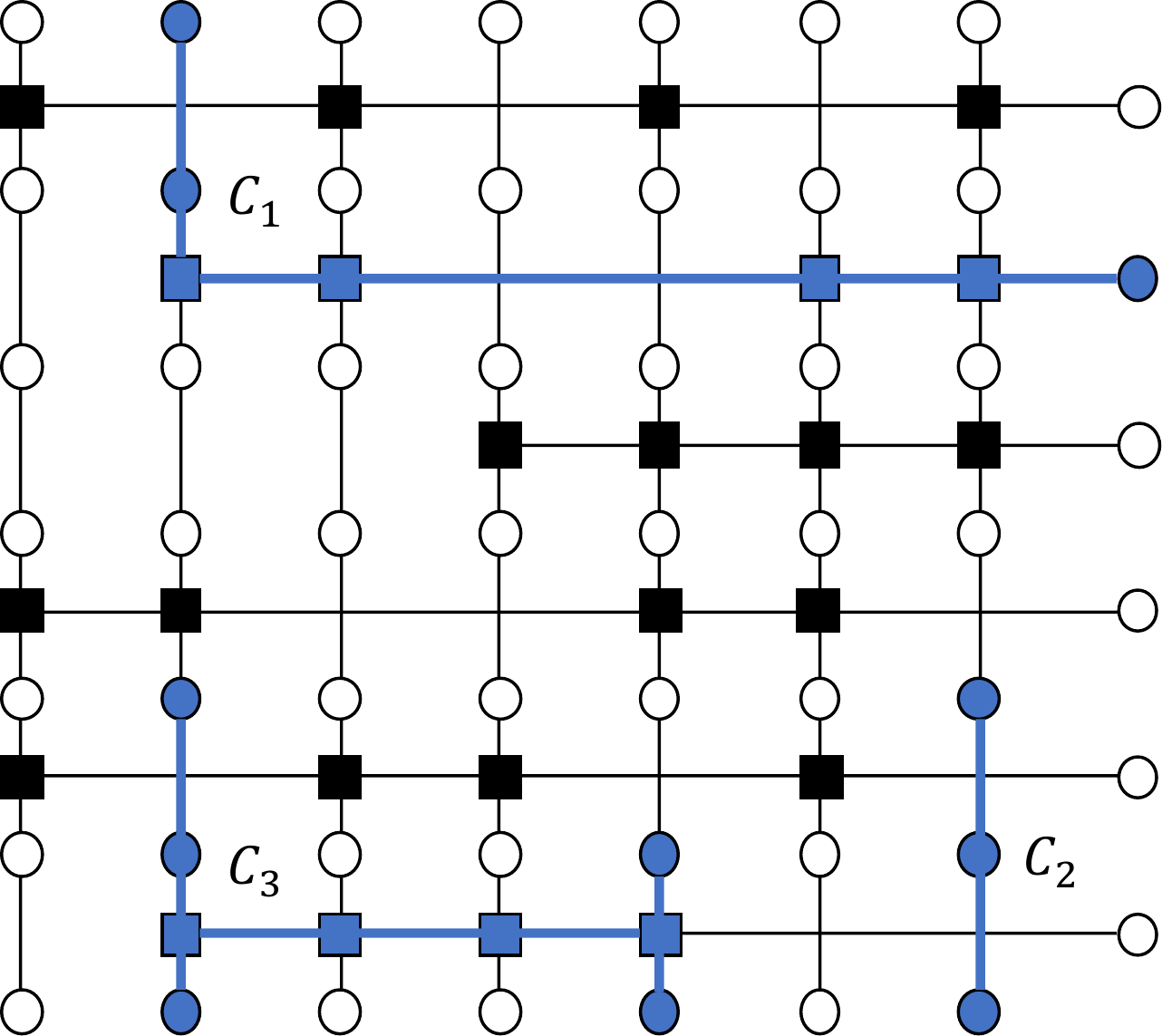}
\caption{A circuit error (red) and the corresponding accumulated error 
(blue) for the Hamming code equipped with a measurement code with 
parameters $[6, 3, 3]$.
The accumulated error has three connected components, which induces
a decomposition of the circuit error into three clusters $C_1, C_2$ and $C_3$.
}
\label{fig:tanner2}
\end{figure}

\medskip
The set of nodes $V = V_D \cup V_M$ is built in such a way that 
each vertex corresponds to a coordinate of a circuit error.
This leads to a one-to-one correspondence between circuit errors $\varepsilon$ 
and subsets $U \subseteq V$ of vertices of the sequential Tanner graph.
An error can be considered alternatively as a vector $\varepsilon = (e, f)$ 
or as a subset $V(\varepsilon) \subseteq V$.
The error $\varepsilon$ whose support is given by $U \subseteq V$
is denoted $\varepsilon(U)$.

\medskip
The sequential Tanner graph provides a graphical framework that 
allows to identify some properties of circuit errors.
Some features of the circuit error are easier to read when considering
the {\em accumulated error} $\bar \varepsilon = (\bar e, f)$ 
defined by
\begin{equation} \label{eqn:accumulated_error}
\bar e^i = \sum_{j=0}^{i} e^j
\end{equation}
for all $i = 0, \dots, n_M$.
The error $\bar e^i$ is the accumulation of all data errors that
appear during the first $i$ measurements.
The residual data error introduced in Section~\ref{subsec:FTEC:ft_decoder} 
is given by $\pi(\varepsilon) = \bar e^{n_M}$.

\medskip
The {\em error graph} induced by a circuit error $\varepsilon$ 
is obtained from the vertex set $V(\varepsilon) \subseteq V$ 
by connecting vertices as follows:
(i) Two consecutive nodes $v_{i, j}$ and $v_{i+1, j}$ in the same column are connected. 
(ii) Two nodes $v_{i, j}$ and $v_{i, k}$ involved in the measurement of
$m_{i+1}$, are connected,
(iii) A node $v_{i,j}$ involved in the measurement of $m_{i+1}$ is connected
to the outcome node $u_{i+1}$.

\medskip
This provides a bijection between circuit errors and error graphs
that allows us to apply the language of graph theory to circuit errors.
A circuit error $\varepsilon$ is said to be {\em connected} if the 
subset $V(\varepsilon)$ induces a connected error graph.
An error $\varepsilon'$ is a {\em connected component} of the circuit error 
$\varepsilon$ if $V(\varepsilon')$ is a connected component of the error graph 
induced by $V(\varepsilon)$.

\medskip
The connected components of the accumulated error $\bar \varepsilon$,
defined in Eq.~\eqref{eqn:accumulated_error},
identify bit flips that trigger the same outcomes. This motivates the 
cluster decomposition that we introduce now.
Let 
$$
V(\bar \varepsilon) = \bigcup_{i \in I} \bar V(\bar \varepsilon_i)
$$
be the decomposition of the accumulated error $\bar \varepsilon$ into connected 
components.
Each component $\bar \varepsilon_i$ is the accumulated error of an error 
$\varepsilon_i$ such that $V(\varepsilon_i) \subseteq V(\varepsilon)$.
The {\em cluster decomposition} of a circuit error $\varepsilon$ is the 
decomposition 
$$
\varepsilon = \sum_{i \in I} \varepsilon_i
$$
derived from the decomposition of the accumulated error $\bar \varepsilon$ 
in connected components.
Figure~\ref{fig:tanner2} shows the cluster decomposition of a circuit error.

\medskip
The {\em input} and {\em output vertices} are 
$$
V_{in}  = \{ v_{0,1}, v_{0,2}, \dots, v_{0, n_D} \} \text{ and } V_{out}  = \{ v_{n_M,1}, v_{n_M,2}, \dots, v_{n_M, n_D} \} 
$$

\medskip
The following lemma justifies the cluster decomposition.

\begin{lemma} \label{lemma:tanner_graph_error}
Let $\varepsilon = \sum_{i \in I} \varepsilon_i$ be the cluster decomposition of a 
circuit error.
\begin{itemize}
\item If $m(\varepsilon) = 0$ then for all $i \in I$ we have $m(\varepsilon_i) = 0$.
\item If $\bar \varepsilon_i \cap V_{out} = \emptyset$ then we have $\pi(\varepsilon_i) = 0$.
\end{itemize}
\end{lemma}

\begin{proof}
The first item holds because, by construction of the sequential Tanner graph, 
the outcome location $m_j$ is connected to all the bits involved in the measurement $m_j$.
The second item is an immediate application of the definition of the 
accumulated error because $\pi(\varepsilon)$ is equal to the accumulated 
error $\bar e^{n_M}$.
\end{proof}

\medskip
The graphical formalism introduced in this section provides a
decomposition of circuit errors and Lemma~\ref{lemma:tanner_graph_error} 
identifies the clusters that contributes to the residual data error.

\subsection{Correction of input error and circuit distance} \label{subsec:FTDECODING:circuit_distance}

The fault tolerance condition introduced in Definition~\ref{def:fault_tolerant_decoder}
can be interpreted as the fact that the decoder corrects the input error
without amplifying internal errors.
This section deals with the correction of the input error ignoring the problem of 
error amplification. 
We introduce a notion of minimum distance $d_{circ}$ adapted to 
the context of fault tolerance and we prove that we can correct the 
input error for any circuit error of weight at most $(d_{circ}-1)/2$.
In Section~\ref{subsec:FTDECODING:truncated_decoder}, we
adapt the decoder in order to keep error 
amplification limited and to satisfy the fault tolerance condition.

\medskip
Given an outcome $m$, denote by 
$\hat \varepsilon = \tilde D_{\MWE}(m)$
a minimum weight circuit error with outcome $m$.
We consider a MWE decoder that returns an estimation
$D_{\MWE}(m) = \pi(\hat \varepsilon)$
of the residual error

\medskip
Naively, for a circuit error $\varepsilon = (e, f)$ with outcome $m(\varepsilon) = m$, 
one could say that the input error is corrected by the MWE decoder if the estimation 
$\hat \varepsilon = (\hat e, \hat f)$ 
satisfies 
$
\hat e^0 = e^0,
$
that is if the input component $e^0$ is correctly estimated. 
This definition is not satisfying because some input errors may be 
indistinguishable from internal errors.
To clarify this point, we introduce the set of trivial errors.
A {\em trivial circuit error} is a circuit error $\varepsilon$ such that 
$m(\varepsilon) = 0$ and $\pi(\varepsilon) = 0$.
This error is impossible to detect since the corresponding
outcome is trivial and it does not induce any bit flip on the data
at the end of the measurement circuit. 
Two circuit errors that differ in a trivial error cannot be distinguished
using the outcome observed or the data bits after measurement.
This notion of equivalence can be seen as a special case of the 
gauge equivalence introduced by Bacon {\em et al.} \cite{bacon2017:QLDPC_from_circuit}
in order to design quantum LDPC codes from a quantum circuit.

\begin{figure}
\centering
\includegraphics[scale=.5]{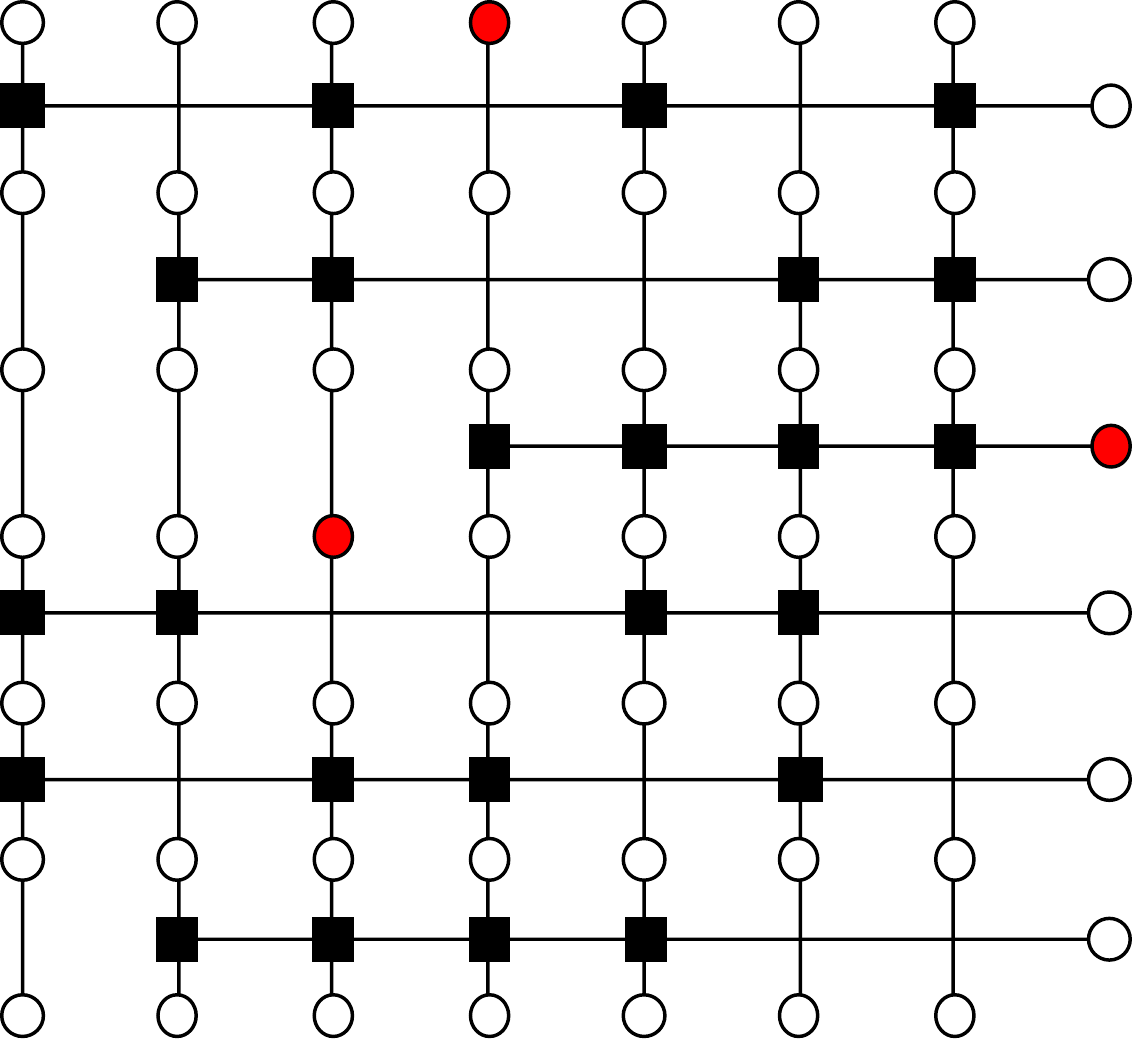}
\hspace{1cm}
\includegraphics[scale=.5]{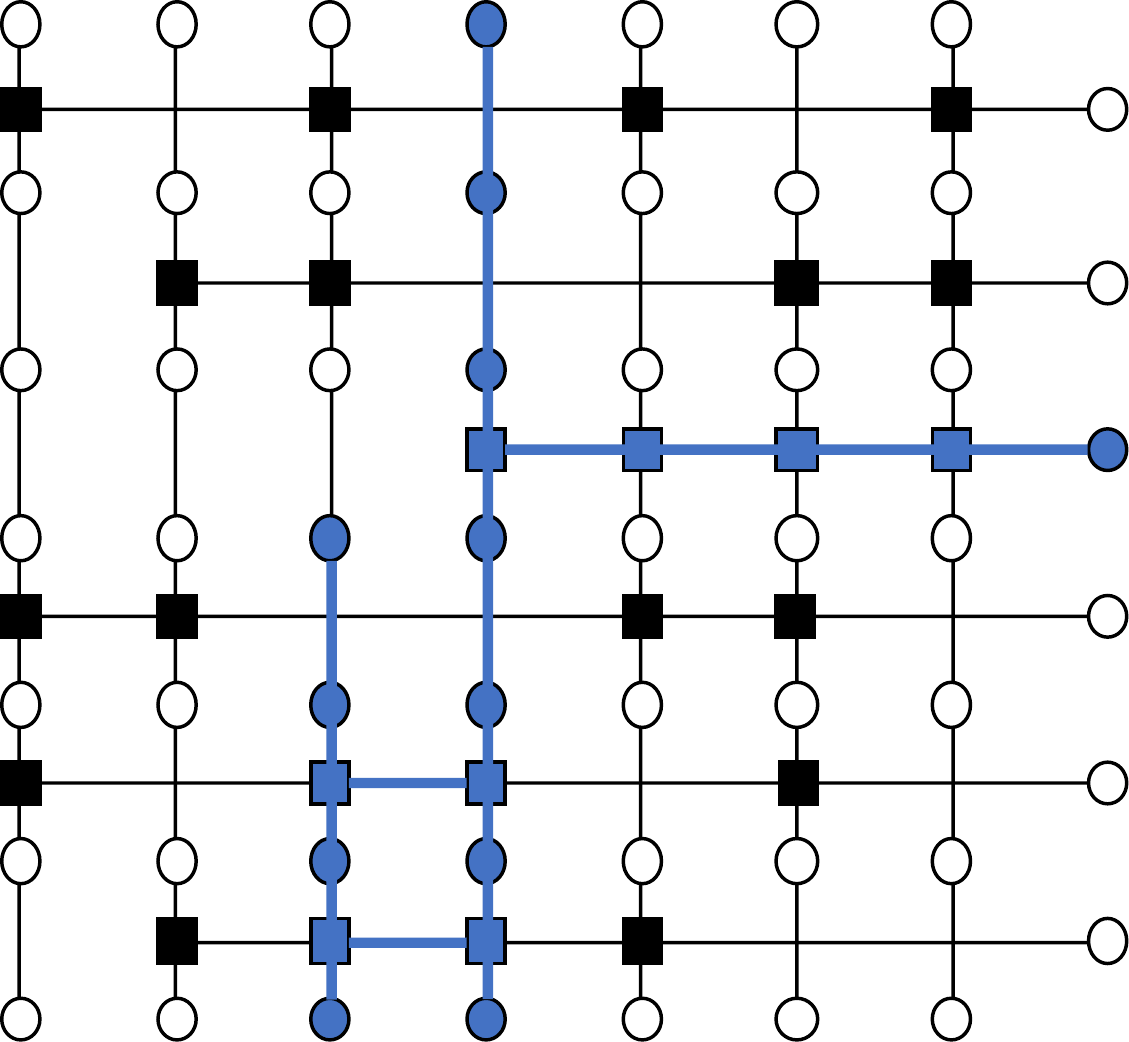}
\caption{A minimum weight circuit error (red) and the corresponding accumulated error (blue) for the Hamming code equipped with a measurement code with parameters $[6, 3, 3]$.
The accumulated error connects an input vertex (top row) with 
an output vertex (bottom row), which means that the circuit error is a propagating error.
Note that the third measurement does not detect the bit flip on the fourth bit 
because of the measurement error.
}
\label{fig:propagating_error}
\end{figure}

\medskip
Our definition of the correction of the input error relies on the notion of 
propagating error that we introduce now.
A {\em propagating error} is defined to be a circuit error $\varepsilon$ with trivial 
outcome $m(\varepsilon) = 0$ such that $V(\bar \varepsilon)$ contains a path connecting 
$V_{in}$ and $V_{out}$.
It can be interpreted as an input error that propagates through the 
measurement circuit without being detected.
Figure~\ref{fig:propagating_error} shows a propagating error
for the Hamming code.
If an error $\varepsilon$ occurs with outcome $m$,
we say that the MWE decoder {\em corrects the input error} 
if $\varepsilon + \tilde D_{\MWE}(m)$ is not a propagating error.
Since this circuit error is guaranteed to have a trivial outcome,
that means that it does not connect input and output sets of vertices.

\medskip
The {\em circuit distance} $d_{circ}$ is defined to be the minimum weight of 
a propagating error.
$$
d_{circ} = \min \{ |\varepsilon|  \text{ such that } \varepsilon \text{ is propagating} \} \cdot
$$
A propagating error is undetectable in the sense that $m(\varepsilon) = 0$ 
and non-trivial, however all undetectable non-trivial errors are not propagating errors.
For instance, the circuit distance of the Hamming code combined with the $[6, 3, 3]$ 
measurement code is three.
A minimum weight propagating error is represented in 
Figure~\ref{fig:propagating_error}.

\medskip
We recalled in Section~\ref{subsec:FTEC:error_correction} that in the standard coding 
theory setting the minimum distance provides an indication on the performance 
of the minimum weight error decoder. Any set of up to $(d-1)/2$ bit flips can be corrected
by MLE decoding.
The following proposition establishes a fault-tolerant analog of this result.

\begin{prop} \label{prop:MLE_dcirc}
For any circuit error $\varepsilon$ such that 
$
|\varepsilon| \leq (d_{circ}-1)/2
$
the MWE decoder $D_{\MWE}$ corrects the input error.
\end{prop}

\begin{proof}
Assume that a circuit error $\varepsilon$ with weight $|\varepsilon| \leq (d_{circ}-1)/2$ 
occurs.
The MWE decoder is based on the estimation $\hat \varepsilon = \tilde D_{\MWE}(m(\varepsilon))$ 
of the circuit error $\varepsilon$.
By definition, it satisfies $|\hat \varepsilon| \leq |\varepsilon| \leq (d_{circ}-1)/2$,
which implies $|\varepsilon + \hat \varepsilon| \leq d_{circ} - 1$.
This proves that the residual circuit error $\varepsilon + \hat \varepsilon$ cannot 
be a propagating error.
The input error is corrected by the MWE decoder.
\end{proof}

\medskip
The circuit distance cannot be arbitrarily large. 
It is limited by the minimum distance $d_D$ of the data code 
and the minimum distance $d_M$ of the measurement code as
$$
d_{circ} \leq \min \{ d_D, n_D + d_M \} \cdot
$$
Indeed, to obtain the upper bound $d_{circ} \leq d_D$ remark that for any 
codeword $u \in C_D$, the circuit error $\varepsilon = (e, f)$ with
input $e^0 = u$ and with $e^1 = \dots = e^{n_M} = f = 0$ is a 
propagating error.
One can also build a propagating error out of an arbitrary input 
error $e^0$ using $f = e^0 H_m^T$.
The second upper bound $d_{circ} \leq n_D + d_M$ follows.

\medskip
Given a data code $C_D$, one can try to select a measurement code
$C_M$ with optimal circuit distance $d_{circ} = d_D$ that requires
a minimum number of parity check measurements $n_M$.
We obtain a circuit distance $d_{circ} = d_D = 3$ for the Hamming code
using the linear codes $[6, 3, 3]$ or $[10, 3, 5]$ defined in 
Eq.~\eqref{eqn:generator_matrices} as a measurement code.
The circuit distance can be larger than the measurement code minimum distance.
The linear code $[5, 3, 2]$ with generator matrix 
$$
G = 
\begin{pmatrix}
1 & 0 & 0 & 1 & 1 \\
0 & 1 & 0 & 1 & 0 \\
0 & 0 & 1 & 1 & 0
\end{pmatrix}
$$
leads to a circuit distance $d_{circ}=3$ for the Hamming code and it requires 
only 5 measurements.

\medskip
As a second example, consider using as data code the BCH code $[15, 7, 5]$ with generator
matrix 
$$
G_D = 
\left(
\begin{array}{c c c c c c c c c c c c c c c}
1 & 0 & 0 & 0 & 0 & 0 & 0 & 1 & 0 & 0 & 0 & 1 & 0 & 1 & 1 \\
0 & 1 & 0 & 0 & 0 & 0 & 0 & 1 & 1 & 0 & 0 & 1 & 1 & 1 & 0 \\
0 & 0 & 1 & 0 & 0 & 0 & 0 & 0 & 1 & 1 & 0 & 0 & 1 & 1 & 1 \\
0 & 0 & 0 & 1 & 0 & 0 & 0 & 1 & 0 & 1 & 1 & 1 & 0 & 0 & 0 \\
0 & 0 & 0 & 0 & 1 & 0 & 0 & 0 & 1 & 0 & 1 & 1 & 1 & 0 & 0 \\
0 & 0 & 0 & 0 & 0 & 1 & 0 & 0 & 0 & 1 & 0 & 1 & 1 & 1 & 0 \\
0 & 0 & 0 & 0 & 0 & 0 & 1 & 0 & 0 & 0 & 1 & 0 & 1 & 1 & 1
\end{array}
\right) \cdot
$$
Searching over random generator matrices $G_M$, we found a
measurement code with length $n_M = 16$ that leads to an optimal circuit 
distance $d_{circ} = d_D = 5$. It is defined by the generator matrix 
$$
G_M = 
\left(
\begin{array}{c c c c c c c c c c c c c c c c}
1 & 0 & 0 & 0 & 1 & 1 & 1 & 0 & 0 & 0 & 0 & 1 & 1 & 0 & 0 & 0 \\
1 & 1 & 1 & 0 & 0 & 1 & 0 & 0 & 1 & 0 & 0 & 0 & 0 & 1 & 1 & 0 \\
1 & 1 & 1 & 1 & 1 & 0 & 1 & 0 & 0 & 1 & 0 & 0 & 0 & 1 & 0 & 1 \\
1 & 0 & 0 & 1 & 0 & 0 & 1 & 0 & 0 & 1 & 0 & 0 & 1 & 0 & 1 & 1 \\
0 & 0 & 1 & 0 & 0 & 1 & 0 & 1 & 1 & 0 & 1 & 1 & 1 & 1 & 0 & 0 \\
1 & 1 & 1 & 1 & 1 & 1 & 0 & 1 & 1 & 0 & 1 & 1 & 0 & 0 & 0 & 1 \\
1 & 0 & 0 & 0 & 0 & 0 & 1 & 0 & 1 & 1 & 0 & 1 & 0 & 0 & 1 & 0 \\
1 & 0 & 1 & 1 & 0 & 0 & 0 & 1 & 1 & 0 & 1 & 0 & 1 & 1 & 0 & 0
\end{array}
\right)
$$

\subsection{Truncated Minimum Weight Error decoder}  
\label{subsec:FTDECODING:truncated_decoder}

We saw that the MWE decoder can be generalized to the context 
of circuit errors by selecting a circuit error $\varepsilon$ with minimum weight 
that yields the observed outcome $m$. 
Then $\pi(\varepsilon)$ provides an estimation of the residual data error that occurs.
Unfortunately, this strategy fails to satisfy the fault tolerance condition of Def.~\ref{def:fault_tolerant_decoder} 
due to the issue of error amplification illustrated by Lemma~\ref{lemma:late_errors_amplification}.
Some internal errors occur too late to be corrected safely.
This motivates the introduction of the {\em truncated} minimum weight error 
decoder.

\begin{figure}
\centering
\includegraphics[scale=.5]{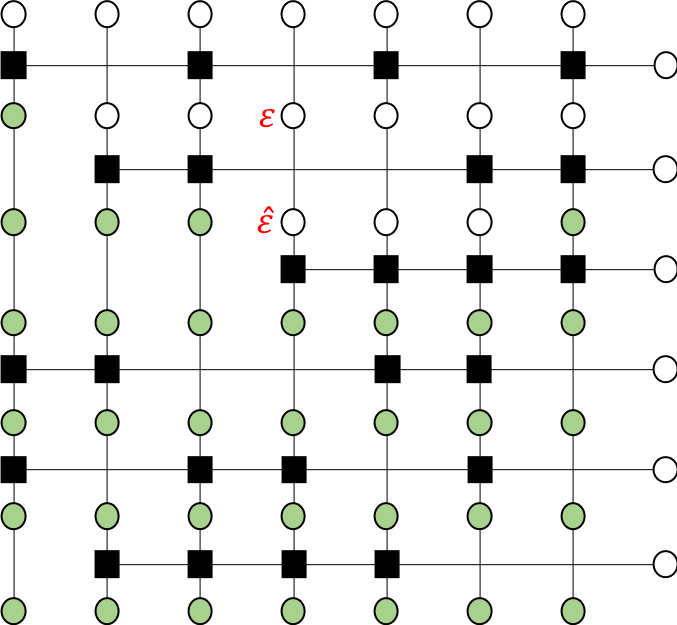}
\hspace{1cm}
\includegraphics[scale=.5]{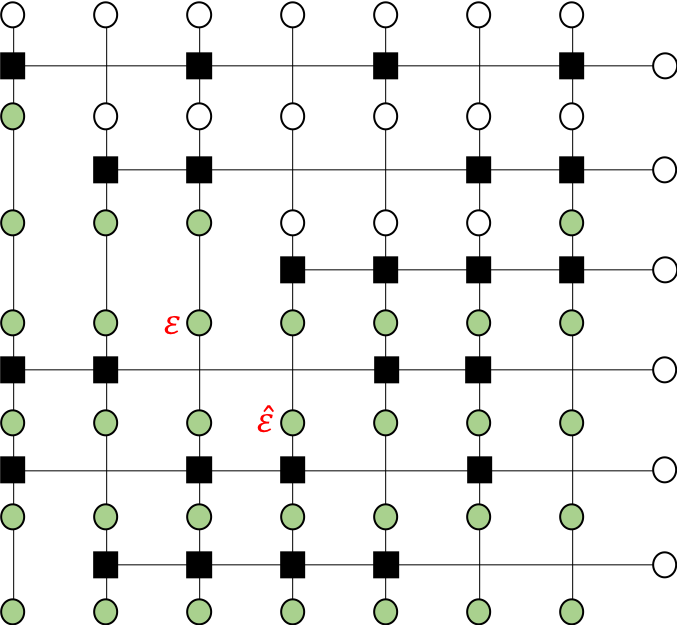}
\caption{Left: A weight-one error $\varepsilon$ and its estimation $\hat \varepsilon$. 
The correction succeeds although the circuit error is not exactly identified since 
no residual error remain at the end of the measurement cycle.
Right: The decoder fails leading to an amplified residual error with weight two. 
To make the MWE decoder fault-tolerant, we will discard the part 
of $\hat \varepsilon$ that is included in the green region $S_{out}$.
}
\label{fig:mwe_decoder_failure}
\end{figure}

\medskip
In order to make the definition of the truncated decoder more intuitive, 
we begin with a case of failure of the minimum weight error decoder illustrated
with Figure~\ref{fig:mwe_decoder_failure}.
An internal bit flip may be amplified by the decoder if it is included in the support of
a weight-two undetectable error with a non-trivial residual error.
To avoid error amplification, we will correct $\varepsilon$ with the restriction 
of $\hat \varepsilon$ to a subset of early bit flip locations. 
In the rest of this section, we determine the exact shape of the restriction. 

\medskip
Let $A \subseteq V$ be a subset of vertices of the sequential Tanner graph.
Let
$
\tilde D_{\MWE}^A
$
be the map defined by $\tilde D_{\MWE}^A(m) = \hat \varepsilon \cap A$ where 
$\hat \varepsilon$ is a minimum weight circuit error with outcome $m$.
We use the notation $\hat \varepsilon \cap A$ as a shorthand for the restriction of the support
of $\varepsilon$ to the set $A$ that is $\hat \varepsilon \cap A = \varepsilon( V(\hat \varepsilon) \cap A)$.
The {\em truncated $\MWE$ decoder} with support $A$ is defined to be
the map 
$
D_{\MWE}^A: \Z_2^{n_M} \rightarrow \Z_2^{n_D}
$
such that 
$$
D_{\MWE}^A(m) = \pi ( \tilde D_{\MWE}^A(m) ) \cdot
$$
For $A = V$, we recover the strategy considered in the previous section,
that is $\tilde D_{\MWE}^V = \tilde D_{\MWE}$.
In the general case, the truncated decoder ignores the bit flips supported 
outside of the subset $A$.
Without loss of generality, we can assume that 
$\tilde D_{\MWE}^V(m) = \hat \varepsilon$ is fixed and that 
$\tilde D_{\MWE}^A(m) = \tilde D_{\MWE}^V \cap A$
for any subset $A$ of $V$.

\begin{figure}
\centering
\includegraphics[scale=.5]{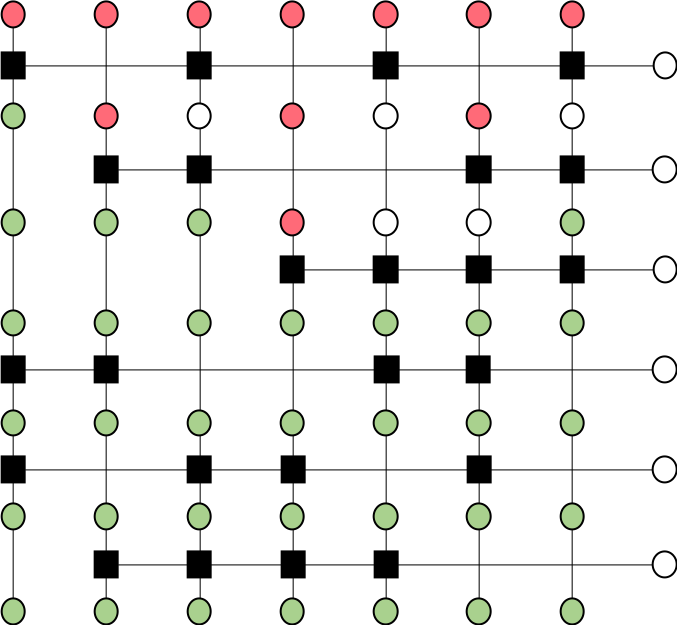}
\caption{The set $V_{in} \cup S_{in}$ (red) and $S_{out}$ (green) do not overlap for the
Hamming code combined with the measurement code $[6,3,3]$.
The restriction of the MWE decoder to the set $A = S_{out}^C$ is fault-tolerant 
by Theorem~\ref{theo:truncated_MWE}.
}
\label{fig:tanner_in_out}
\end{figure}

\medskip
Let $S_{in} \subseteq V$ be the union of the supports of all connected circuit errors $\varepsilon$
with weight up to $d_{D}-1$ such that $\varepsilon \cap V_{in} \neq \emptyset$.
Define $S_{out} \subseteq V$ as the union of the supports of all connected circuit errors $\varepsilon$
with weight up to $d_{D}-1$ such that $\bar \varepsilon \cap V_{out} \neq \emptyset$.

\begin{theo} \label{theo:truncated_MWE}
If $( V_{in} \cup S_{in} ) \cap S_{out} = \emptyset$, then the truncated decoder
$D_{\MWE}^A$ with $A = S_{out}^C$ is a fault-tolerant decoder.
\end{theo}

In what follows, when we refer to the truncated MWE decoder, we assume that 
the support of the truncated decoder is $A = S_{out}^C$.
The condition $V_{in} \cap S_{out} = \emptyset$ is equivalent to $d_{circ} = d_D$.
A large circuit distance is therefore required in order to ensure fault tolerance.

\begin{proof}
Consider an error $\varepsilon = (e, f)$ with outcome $m$ such that 
$|\varepsilon| \leq (d_D-1)/2$ and denote by
$$
\hat \pi^A = \pi( \hat \varepsilon \cap A )
$$
the residual error estimation returned by the truncated MWE decoder 
where $A = S_{out}^C$.

\medskip
We are interested in the residual data error after correction,
{\em i.e.}
$$
\pi(\varepsilon) + \hat \pi^A = \pi( \varepsilon + \hat \varepsilon \cap A ) \cdot
$$
Let us prove that it satisfies the fault tolerance condition
$
| \pi( \varepsilon + \hat \varepsilon \cap A )|  \leq |\varepsilon| - |e^0|.
$

\medskip
{\bf Partition of the circuit error:}
Denote $\omega = \varepsilon + \hat \varepsilon$.
The set $V(\omega)$ is the set of locations where $\varepsilon$
and its estimation $\hat \varepsilon$ do not match.
We will prove the fault tolerance condition in two steps through the partition 
$V = V(\omega) \cup V(\omega)^C$.

\medskip
Denote 
$\varepsilon \cap V(\omega) = (e_1, f_1)$ 
and 
$\varepsilon \cap V(\omega)^C = (e_2, f_2)$ 
the two components of $\varepsilon$.
It is enough to show that both components satisfy the fault tolerance constraint, 
that is
\begin{align}
| \pi( \varepsilon \cap V(\omega) + \hat \varepsilon \cap V(\omega) \cap A )| \leq |\varepsilon \cap V(\omega)| - |e_1^0| \label{eq:theo_proof_cases1}
\end{align} 
and
\begin{align}
| \pi( \varepsilon \cap V(\omega)^C + \hat \varepsilon \cap V(\omega)^C \cap A )| \leq |\varepsilon \cap V(\omega)^C| - |e_2^0| \label{eq:theo_proof_cases2}
\end{align} 
Assuming that
Eqs.~\eqref{eq:theo_proof_cases1} and \eqref{eq:theo_proof_cases2} are satisfied, we obtain
the fault tolerance condition as follows:
\begin{align*}
| \pi( \varepsilon + \hat \varepsilon \cap A )|
	& = | \pi( (\varepsilon + \hat \varepsilon \cap A) \cap V(\omega) + (\varepsilon + \hat \varepsilon \cap A) \cap V(\omega)^C )| \\
	& \leq | \pi( (\varepsilon + \hat \varepsilon \cap A) \cap V(\omega) )| + |\pi( (\varepsilon + \hat \varepsilon \cap A) \cap V(\omega)^C )| \\
	& = | \pi( \varepsilon \cap V(\omega) + \hat \varepsilon \cap V(\omega) \cap A )|  
		+  | \pi( \varepsilon \cap V(\omega)^C + \hat \varepsilon \cap V(\omega)^C \cap A )| \\
	& \leq |\varepsilon \cap V(\omega)| - |e_1^0| + |\varepsilon \cap V(\omega)^C| - |e_2^0| \\
	& = |\varepsilon| - |e^0|
\end{align*}
In the remainder of the proof, we demonstrate Eq.~\eqref{eq:theo_proof_cases1} 
and Eq.~\eqref{eq:theo_proof_cases2}.

\medskip
{\bf Proof of Eq.~\eqref{eq:theo_proof_cases2}:}
By definition, the set 
$V(\omega)^C$ is the subset of $V$ over which $\varepsilon$ and $\hat \varepsilon$ 
coincide, {\em i.e.} 
$\varepsilon \cap V(\omega)^C = \hat \varepsilon \cap V(\omega)^C$. 
Consequently,
\begin{align*}
\varepsilon \cap V(\omega)^C + \hat \varepsilon \cap V(\omega)^C \cap A 
	& = \varepsilon \cap V(\omega)^C \cap A^C
\end{align*}
which produces
\begin{align*}
| \pi( \varepsilon \cap V(\omega)^C + \hat \varepsilon \cap V(\omega)^C \cap A ) | 
	& \leq | \pi( \varepsilon \cap V(\omega)^C \cap A^C ) | \\
	& \leq | \varepsilon \cap V(\omega)^C \cap A^C | \\
	& \leq |\varepsilon \cap V(\omega)^C| - |e_2^0|
\end{align*}
where the last inequality exploits the fact that $A^C = S_{out}$
does not intersect $V_{in}$.
This proves Eq.~\eqref{eq:theo_proof_cases2}.

\medskip
{\bf Proof of Eq.~\eqref{eq:theo_proof_cases1}:}
Consider the cluster decomposition $\omega = \sum_{i \in I} \omega_i$ of $\omega$
and denote by 
$\varepsilon_i = \varepsilon \cap V(\omega_i)$ 
and 
$\hat \varepsilon_i = \hat \varepsilon \cap V(\omega_i)$.
Since $\omega = \varepsilon + \hat \varepsilon$, we have $\omega_i = \varepsilon_i + \hat \varepsilon_i$.
From Lemma~\ref{lemma:tanner_graph_error}, $m(\omega_i) = 0$ 
for each cluster since $m(\omega) = 0$.
The clusters also satisfy $|\omega_i| \leq |\omega| \leq d_D-1$ as required
in the definition of $S_{in}$ and $S_{out}$.
The cluster decomposition leads to 
\begin{align} \label{eqn:bound_amplification}
\pi( \varepsilon \cap V(\omega) + \hat \varepsilon \cap V(\omega) \cap A )
	& = \sum_{i \in I} \pi(\varepsilon_i + \hat \varepsilon_i \cap A) 
\end{align}
by linearity of $\pi$.
The term $\pi(\varepsilon_i + \hat \varepsilon_i \cap A)$ depends on the relative position
of the error $\omega_i$ and the truncated set $A$. 
We will establish the fault tolerance inequality for each term 
$\varepsilon_i + \hat \varepsilon_i \cap A$ by considering three cases 
as illustrated with Figure~\ref{fig:theo_proof}.

\begin{figure}
\centering
\includegraphics[scale=.5]{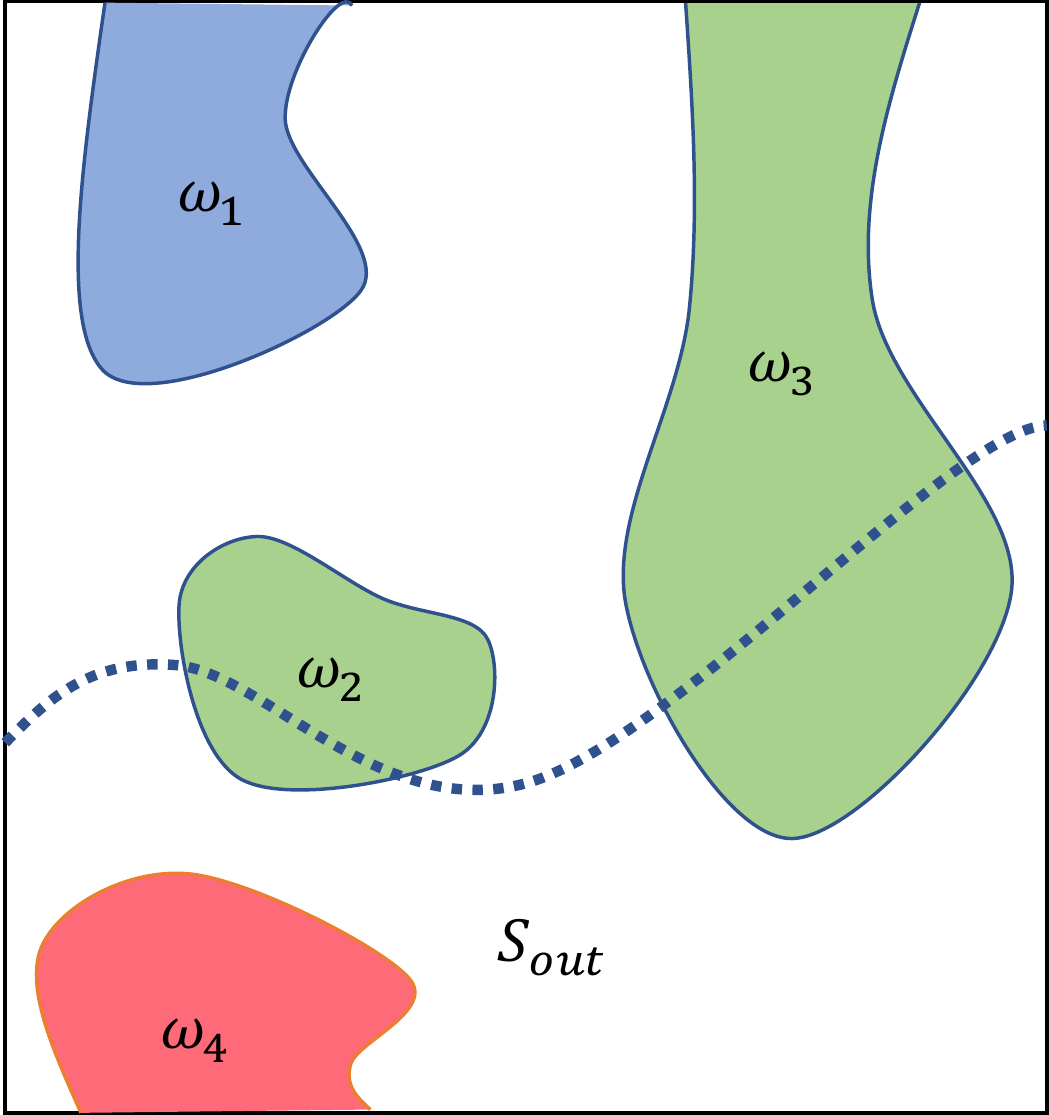}
\caption{The three types of configurations for the cluster of $\omega$ 
in Step 3 of the proof of Theorem~\ref{theo:truncated_MWE}.
The clusters of $\omega$ included in the top region like $\omega_1$
appear early enough to be corrected (case (a)).
The cluster $\omega_4$ which is fully included in $S_{out}$ and it is 
entirely truncated because it appears too late to be corrected (case (b)).
The clusters that overlap with both $S_{out}$ and its complementary
like $\omega_2$ and $\omega_3$ are partly corrected (case (c)). 
The only input error that contributes to the residual error after correction
belongs to the cluster $\omega_3$. The assumption of Theorem~\ref{theo:truncated_MWE}
guarantees that such a cluster cannot exist.
}
\label{fig:theo_proof}
\end{figure}

\begin{enumerate}
\item [(a)] Assume first that $\omega_i \subseteq A = S_{out}^C$.
Then, we have 
$\varepsilon_i + \hat \varepsilon_i \cap A = \varepsilon_i + \hat \varepsilon_i = \omega_i$. 
The accumulated error $\bar \omega_i$ cannot intersect $V_{out}$ otherwise 
it would included in $S_{out}$. Hence Lemma~\ref{lemma:tanner_graph_error} tells us that 
\begin{align}  \label{eqn:bound_amplification_case1}
\pi(\varepsilon_i + \hat \varepsilon_i \cap A ) = \pi(\omega_i) = 0 \cdot
\end{align}
\item [(b)]  Consider now the case $\omega_i \subseteq A^C = S_{out}$.
Then, we have $\varepsilon_i + \hat \varepsilon_i \cap A = \varepsilon_i$ that yields
\begin{align}  \label{eqn:bound_amplification_case2}
| \pi(\varepsilon_i + \hat \varepsilon_i \cap A ) |
= |\pi(\varepsilon_i)| \leq |\varepsilon_i| \cdot
\end{align}
\item [(c)]  The remaining clusters $\omega$ intersect both $A$ and its complementary 
$A^C$.
By definition of $S_{out}$, such an error $\omega_i$ cannot meet $V_{out}$ otherwise
it would be fully included in $A^C = S_{out}$.
One can thus apply Lemma~\ref{lemma:tanner_graph_error} showing
that $\pi(\omega_i) = \pi(\varepsilon_i + \hat \varepsilon_i) = 0$.
This leads to
\begin{align}  \label{eqn:bound_amplification_case3}
|\pi(\varepsilon_i + \hat \varepsilon_i \cap A)| = |\pi(\hat \varepsilon_i \cap A^C)| \leq |\hat \varepsilon_i| \leq |\varepsilon_i| \cdot
\end{align}
Therein, the last inequality is a consequence of Lemma~\ref{lemma:local_MWE} below.
\end{enumerate}
Denote by 
$I_{(a)} = \{ i \ | \ \omega_i \subseteq A \}$, 
$I_{(b)} = \{ i \ | \ \omega_i \subseteq A^C \}$ 
and
$I_{(c)} = I \backslash ( I_{(a)} \cup I_{(b)} )$,
the index sets corresponding to the previous three cases.
Injecting the three inequalities 
\eqref{eqn:bound_amplification_case1}, 
\eqref{eqn:bound_amplification_case2} and
\eqref{eqn:bound_amplification_case3}  
in Eq.~\eqref{eqn:bound_amplification} 
leads to
\begin{align} \label{eqn:bound_amplification2}
\pi( \varepsilon \cap V(\omega) + \hat \varepsilon \cap V(\omega) \cap A ) \leq \sum_{i \in I_{(b)} \cup I_{(c)}} |\varepsilon_i| \cdot
\end{align}

\medskip
To show Eq.~\eqref{eq:theo_proof_cases1}, it remains to prove that this sum is at most $|\varepsilon \cap V(\omega)| - |e_1^0|$.
Consider the error $\omega_{in} = \sum_{i \in I_{in}} \omega_i$ which is the sum of all
the clusters of $\omega$ that intersect with $V_{in}$.
The input error $e^0$ of $\varepsilon$ is included in the support
of $\sum_{i \in I_{in}} \varepsilon_i$.
By definition, if $i \in I_{in}$ then $\omega_i \subseteq S_{in}$.
Using the hypothesis $S_{in} \cap S_{out} = \emptyset$
this proves that $\omega_i \subseteq S_{in} \subseteq S_{out}^C = A$.
This shows that $I_{in} \subseteq I_{(a)}$ and thus $ I_{(b)} \cup  I_{(c)} \subseteq I_{in}^C$.
Coming back  to Eq.~\eqref{eqn:bound_amplification2}, we obtain
$$
\pi( \varepsilon \cap V(\omega) + \hat \varepsilon \cap V(\omega) \cap A )  \leq \sum_{i \in I_{(b)} \cup I_{(c)}} |\varepsilon_i|  \leq \sum_{i \in I_{in}^C} |\varepsilon_i| \leq |\varepsilon \cap V(\omega)| - |e_1^0|
$$
concluding the proof of Eq.~\eqref{eq:theo_proof_cases1}.
The Theorem follows.
\end{proof}

Consider an error $\varepsilon$ with outcome $m$ and let $\hat \varepsilon = D_{\MWE}^V(m)$.
The following lemma proves that a minimum weight error $\hat \varepsilon$ is also 
locally minimum within each cluster of $\varepsilon + \hat \varepsilon$. 
\begin{lemma} \label{lemma:local_MWE}
Let $\varepsilon$ be a circuit error with outcome $m$, let $\hat \varepsilon = \tilde D_{\MWE}^V(m)$
and let $\omega = \varepsilon + \hat \varepsilon$. 
Denote by $\omega = \sum_{i} \omega_i$ the cluster decomposition of $\omega$
and let $\varepsilon_i = \varepsilon \cap V(\omega_i)$ and $\hat \varepsilon_i = \hat \varepsilon \cap V(\omega_i)$.
Then, for all $i$, we have $|\varepsilon_i| \geq |\hat \varepsilon_i|$.
\end{lemma}

\begin{proof}
If there exists a cluster $i$ such that $|\varepsilon_i| < |\hat \varepsilon_i|$ then replacing 
$\hat \varepsilon_i$ by $\varepsilon_i$ in $\hat \varepsilon$
provides an error $\hat \varepsilon' = \hat \varepsilon + \varepsilon_i + \hat \varepsilon_i$ 
with reduced weight and unchanged outcome
$m(\hat \varepsilon') = m(\hat \varepsilon)$. This last equality is a based on the fact that 
$m(\varepsilon_i + \hat \varepsilon_i) = m(\omega_i) = 0$
proven in Lemma~\ref{lemma:tanner_graph_error}.
This cannot happen by definition of the MWE decoder. 
\end{proof}

\section{Time overhead of fault tolerance} \label{sec:BOUND}

The choice of the encoding scheme is driven by the application considered. 
The application dictates the number of data bits $k$ that we need to 
encode and the error rate targeted 
is used to estimate the minimum distance $d$ required.
Encoding increases the volume of the data. 
The {\em space overhead} is the inverse of the rate of the code used, 
{\em i.e.} roughly we need $1/R$ bits per data bit. 
The {\em time overhead} to implement a fault-tolerant error correction scheme
is the number of parity check measurements per correction cycle. 
Fault tolerance may considerably increase the number of measurements 
needed to perform error correction with a code of length $n_D$. 
In this section, we obtain an upper bound on number of measurements 
required to guarantee fault tolerance by analyzing the circuit distance
of random measurement matrices.

\medskip
The following theorem demonstrates the existence of short length fault-tolerant 
sequences for general families of codes.
By a fault-tolerant sequence, we mean a sequence of parity check measurements that 
makes the data code fault-tolerant using the truncated MWE decoder.

\begin{theo} \label{theo:time_overhead}
Let $C_D$ be a family of data codes with minimum distance $d_D$
and length $n_D$.
\begin{itemize}
\item {\bf Polylog distance:} Suppose that $d_D \geq A \cdot \log(n_D)^\alpha$ 
for some constants $A, \alpha > 0$.
There exists a constant $d_0$ such that if $d_D \geq d_0$
the code $C_D$ admits a fault-tolerant measurement sequence with length $n_M = O(d_D^{1 + 1/\alpha})$.
\item {\bf Polynomial distance:} Suppose that $d_D \geq A  \cdot n_D^\alpha$ 
for some constants $A, \alpha > 0$.
There exists a constant $d_0$ such that if $d_D \geq d_0$
the code $C_D$ admits a fault-tolerant measurement sequence with length $n_M = O(d_D \log(n_D))$.
\end{itemize}
\end{theo}

In term of circuit distance, we prove that there exists a family of measurement 
codes with length $n_M$ that produces an optimal circuit distance 
$d_{circ} = d_D$ for the data code $C_D$.

\medskip
Naturally, one can trade time for space. In this context, this can be done by encoding
our $k$ data bits with a longer code $C_D$ with same minimum distance $d_D$. 
This extra cost in space can be compensated with a shorter fault-tolerant
measurement sequence. 
If the distance $d_D$ grows linearly with $n_D$, then the theorem
provides a fault-tolerant sequence of $\Omega(n_D \log(n_D))$ measurements.
However, using a code with minimum distance $d_D = \Omega(n_D^{\beta})$ 
for some $0 < \beta < 1$, only $\Omega(n_D^{\beta} \log(n_D))$ parity check measurements
suffice for fault-tolerance.

\begin{proof}
The basic idea is to build a family of measurement codes $C_M$ that maximizes 
the circuit distance of the pair $(C_D, C_M)$.

In order to guarantee an optimal circuit distance, we must prove that
it is possible to construct a measurement matrix $H_m$ such that
there is no circuit error with weight $w \leq d_D-1$ that is a propagating
error.
We will use the probabilistic method.
Fix the code $C_D$ and pick a random measurement matrix $H_m = G_M^T H_D$
whose rows are $n_M$ vectors of $C_D^\perp$ selected independently
according to a uniform distribution.

For a circuit error $\varepsilon \in \Z_2^N$, define the random variable 
$
X_{\varepsilon}
$
by
$$
X_{\varepsilon}(H_m) = 
\begin{cases}
1 \text{ if } \varepsilon \text{ is a propagating error for } H_m\\
0 \text{ otherwise}
\end{cases}
$$
Note that $N = (n_M + 1) (n_D + 1) - 1$.
Then, for $\rho \in \N$ denote
$$
X_{\rho} = \sum_{\substack{ \varepsilon \in \Z_2^N \\ |\varepsilon| \leq \rho }} X_\varepsilon
$$
the random variable that counts the number of propagating errors with 
weight up to $\rho$ for the code $C_D$ with the measurement matrix $H_m$.
In what follows, $\rho = d_D - 1$ and our goal is to bound the expectation of $X_\rho$.

By definition, the expectation of $X_{\varepsilon}$ is the probability that 
$\varepsilon = (e, f)$ is a propagating error.
Based on Lemma~\ref{lemma:random_code_lemma1} below, this probability is 
upper bounded by the probability that $m(e, 0) = f$.
First, let us prove that the vector $m(e, 0)$ is a uniformly random bit string 
of $\Z_2^{n_M}$. 
For all $i=1, \dots, n_M$, the $i$ th component $m_i$ of $m(e, 0)$ is the inner 
product between row $i$ of $H_m$ and the component $\bar e^{(i-1)}$
of the accumulated error.
Moreover, Lemma~\ref{lemma:random_code_lemma1} shows that 
$\bar e^{(i-1)} \notin C_D$, which proves that $m_i$ is a uniform random bit.
Given that rows of $H_m$ are selected independently, for any
circuit error $x$ with weight $|x| < d_D$, the vector $m(e, 0)$
is uniformly distributed in $\Z_2^{n_M}$.
This produces the upper bound
$$
\Esp(X_\varepsilon) = \Prob( \varepsilon \text{ is a propagating error} ) 
\leq \Prob( m(e, 0) = f ) = 2^{-n_M} 
$$
where the last equality is based on the uniformity of $m(e, 0)$.

Linearity of the expectation, combined with the upper bound on $\Esp(X_\varepsilon)$ 
leads to
$$
\Esp(X_\rho) 
= \sum_{\substack{ \varepsilon \in \Z_2^N \\ |\varepsilon| \leq \rho }} \Esp(X_\varepsilon) 
\leq \rho \cdot \binom{N}{\rho} 2^{-n_M}
\leq d_D \cdot \binom{n_D n_M}{d_D} 2^{-n_M}
$$
where $N = (n_D+1)(n_M+1)-1$ and $\rho = d_D-1$.

Applying Lemma~\ref{lemma:random_code_lemma2}, we get
\begin{align} \label{eq:X_rho_bound}
\Esp(X_\rho) 
	& \leq d_D \cdot 2^{d_D \left( \log_2( e \cdot n_D \cdot  n_M / d_D ) \right) - n_M}
\end{align}

Consider first the polylog distance case. 
We have 
$d_D \geq A \log(n_D)^\alpha$,
or equivalently 
$e^{ B d_D^{1/\alpha}} \geq  n_D$
for some constant $B$.
For a sequence with length $n_M$, this leads to the following exponent in Eq.~\eqref{eq:X_rho_bound}:
\begin{align*}
& d_D \log_2( e \cdot n_D \cdot n_M / d_D ) - n_M \\
& \leq d_D \log_2( e^{ B d_D^{1/\alpha}}) + d_D \log_2( e \cdot n_M / d_D ) - n_M \\
& =C d_D^{1+1/\alpha} +  d_D \log_2( e \cdot n_M / d_D ) - n_M
\end{align*}
for some constant $C$.
One can select a sequence length $n_M = O(d_D^{1+1/\alpha})$
such that this exponent goes to $- \infty$ and therefore $\Esp(X_\rho) \rightarrow 0$
when $d_D \rightarrow +\infty$.

Consider now the polynomial distance case: 
$d_D \geq A n_D^\alpha$,
which means 
$B d_D^{1/\alpha} \geq n_D$ for some constant $B$.
The resulting exponent in Eq.~\eqref{eq:X_rho_bound} is
\begin{align*}
& d_D \log_2( e \cdot n_D \cdot n_M / d_D ) - n_M \\
& \leq d_D \log_2( B d_D^{1/\alpha} ) + d_D \log_2( e \cdot n_M / d_D ) - n_M \\
& = d_D \log_2(B) + C d_D \log_2(d_D) +  d_D \log_2(e) + d_D \log_2( n_M / d_D ) - n_M
\end{align*}
for some constant $C$.
Consider a sequence length $n_M = A' d_D \log_2(d_D)$ for some constant $A'$
such that $A' > \max(1, C)$.
The last term $n_M$ dominates the terms $d_D \log_2(B), d_D \log_2(e)$ and 
$C d_D \log_2(d_D)$.
It remains the term
$
d_D \log_2( n_M / d_D )  
= d_D \log_2( A' ) + d_D \log_2( d_D ) 
$
which is also dominated by $n_m$.
Again, this proves that for $n_M = O(d_D \log(d_D))$, the sequence 
$\Esp(X_\rho)$ goes to $0$ when $d_D \rightarrow +\infty$.

In both cases (polylog and polynomial distance), we showed that
$\Esp(X_\rho)$ goes to $0$.
Since $X_\rho$ takes integer values, this is enough to prove 
the existence of a measurement code family such that $X_{\rho} = 0$
for all sufficiently large $d_D$.
By definition of $X_\rho$, this family has an optimal circuit distance 
for the data codes $C_D$.

To conclude, we apply Theorem~\ref{theo:truncated_MWE}. 
In general, the condition $( V_{in} \cup S_{in} ) \cap S_{out} = \emptyset$,
required to apply the theorem, is not satisfied. But it is sufficient to repeat twice
the measurement sequence to guarantee this condition.
This concludes the proof.
\end{proof}

\begin{lemma} \label{lemma:random_code_lemma1}
If $\varepsilon = (e, f)$ is a propagating error with weight 
$|\varepsilon| < d_D$ then the accumulated error $\bar e_i$
introduced in Eq.~\eqref{eqn:accumulated_error} is not a codeword of $C_D$
and $m(e, 0) = f$.
\end{lemma}
Property (i) of Lemma~\ref{lemma:random_code_lemma1} is independent of 
the codes $C_D$ and $C_M$.
However, the value of $m(e, 0)$ used in (ii) depends on these codes. 

\begin{proof}
By definition of a propagating error, we have 
$\bar e^{i} \neq 0$ for all $i$ and the condition $|\varepsilon| < d_D$ implies
that $\bar e^{i}$ cannot belong to $C_D$. This proves item (i).
The second property is an immediate consequence of the property 
$m(\varepsilon) = 0$
\end{proof}

The proof of Theorem~\ref{theo:time_overhead} relies on the following 
standard bound on combinatorial factors.
\begin{lemma} \label{lemma:random_code_lemma2}
For all integers $m, n$ such that $1 \leq m \leq  n$, we have
$$
\binom{n}{m} \leq \left( \frac{ne}{m} \right)^m  = 2^{m \log_2( n/m ) + m \log_2(e)} \cdot
$$
\end{lemma}

\begin{proof}
It is an immediate application of the bound $m! \geq \left( \frac{m}{e} \right)^m$
\end{proof}

\section{Numerical results} \label{sec:numerics}

This section illustrates our results with numerical simulations. 
As proven in Proposition~\ref{prop:ave_life_time_increase}, 
we observe an increase of the lifetime of encoded data,
corrected regularly using the truncated MWE decoder,
when the initial physical noise rate is sufficiently low.
Then, we analyze the importance of different types of noise by varying the 
relative probabilities of input errors, internal errors and measurement errors,
proving that internal errors are the most harmful.

\medskip
Given a data code $C_D$, we select a measurement matrix with optimal
circuit distance. We pick a length $n_M$ as small as possible.
The truncated MWE decoder is used for fault-tolerant error correction. 
We implement this decoding algorithm as a lookup table.
This strategy applies to a restricted set of codes since the amount of memory required 
to store the table grows exponentially with the code length. 
The main advantage of this approach is the rapidity of the decoding that 
returns the correction to apply in constant time.

\begin{figure}
\centering
\includegraphics[scale=.48]{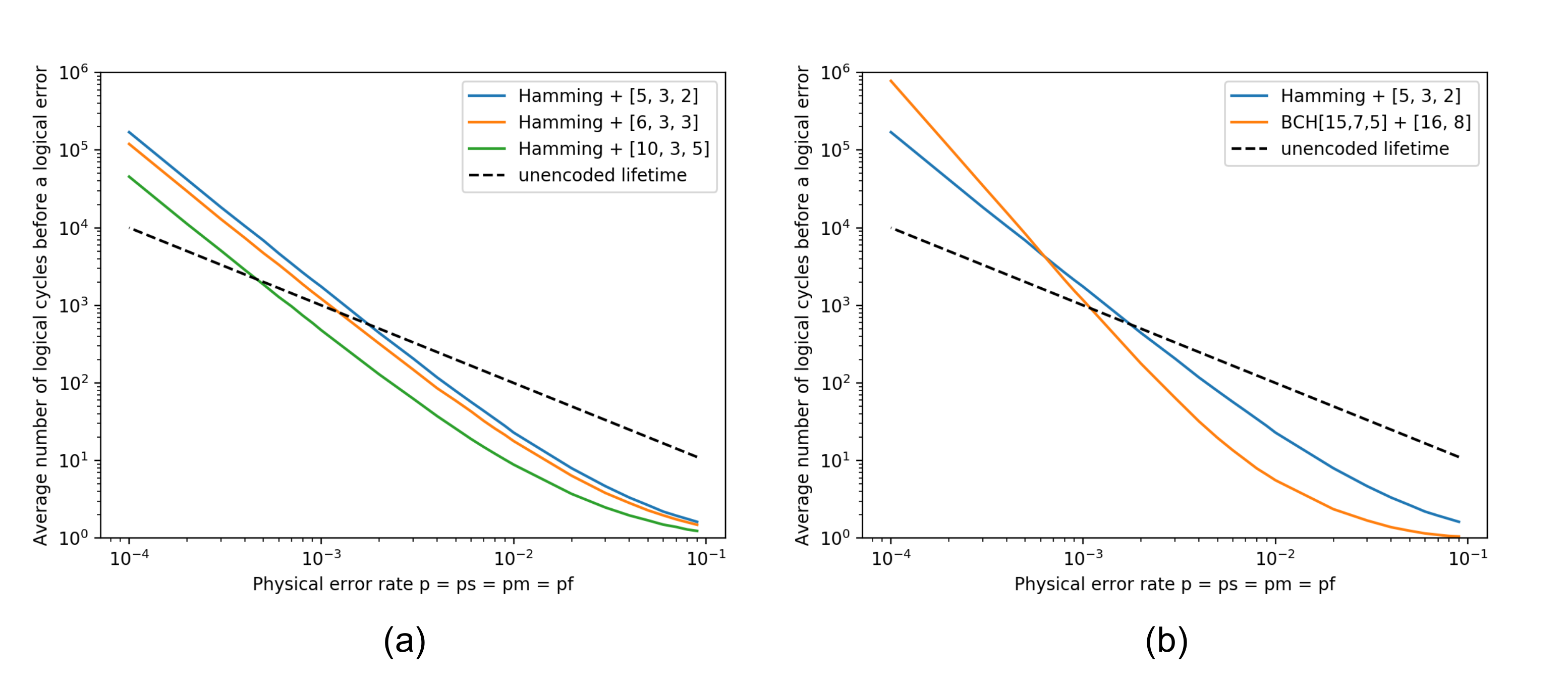}
\caption{(a) Average lifetime of the Hamming code with different measurement codes
for a uniform noise with parameters $p = p_s = p_m = p_f$.
In this noise regime, the shortest measurement sequence with five measurements gives the best results.
(b) Comparison between the Hamming code and a distance-five BCH code. 
The larger minimum distance of the BCH code leads to a more favorable scaling of 
the encoded lifetime.
}
\label{fig:life_time}
\end{figure}

\medskip
Figure~\ref{fig:life_time} plots the average lifetime obtained by numerical simulations.
We assume that we perform cycles of measurement and error correction 
at regular intervals. Between two such error correction cycles, the stored data is affected
by independent by flips with probability $p_s$. We refer to $p_s$ as the storage 
error rate.
During a correction cycle, each parity check measurement may flip the measured bits.
We assume that the noise on the bits that are not involved in the parity check is negligible.
Measured bits are affected by independent bit flip with probability $p_m$.
During a full measurement cycle, a bit involved in $r$ parity checks suffers from
an error rate that is roughly $r p_m$.
Each outcome bit is flipped independently with probability $p_f$.
We estimate the lifetime of encoded data by compute the average lifetime  
over 10000 trials.
When the physical error rate is small enough the lifetime of the encoded data surpasses 
the unencoded lifetime. 
In the case of a uniform noise $p_s = p_f = p_m$, this happens for 
$$
p_{th} \approx 1.1 \cdot 10^{-3}
$$
for the Hamming code combined with the linear code $[5, 3, 2]$.
When the error rate $p$ is below the threshold value $p_{th}$, often called 
{\em pseudo-threshold} \cite{svore2005:pseudothreshold}, it becomes advantageous to encode. For a uniform noise 
a smaller number of measurements, that is smaller length for the measurement code 
is preferable. A larger minimum distance $d_D$ brings a greater improvement 
of the average lifetime below the pseudo-threshold but it generally also 
degrades the value of the pseudo-threshold of the scheme.

\medskip
The average lifetime and the pseudo-threshold of a fault-tolerant error correction 
scheme depends on the three parameters $p_s, p_f, p_m$ of the storage noise model.
Figure~\ref{fig:life_time_multivariate} shows that the parameter $p_m$ has a greater 
influence on the performance of the scheme than the flip error rate $p_f$.
An internal bit flip is more likely to cause a logical error than a flipped outcome.
This is because an error that affects only the measurement outcome leads to 
introducing an error $D_{\MWE}(m)$ in the data and by construction of the 
decoder this error is chosen to have low weight.
Through this process flipped outcomes are converted into low-weight residual 
errors that can be corrected by the next error correction cycle.
This is true even when many outcomes is flipped. 
This phenomenon makes outcome flips far easier to correct than bit flips corrupting
the data.

\begin{figure}
\centering
\includegraphics[scale=.5]{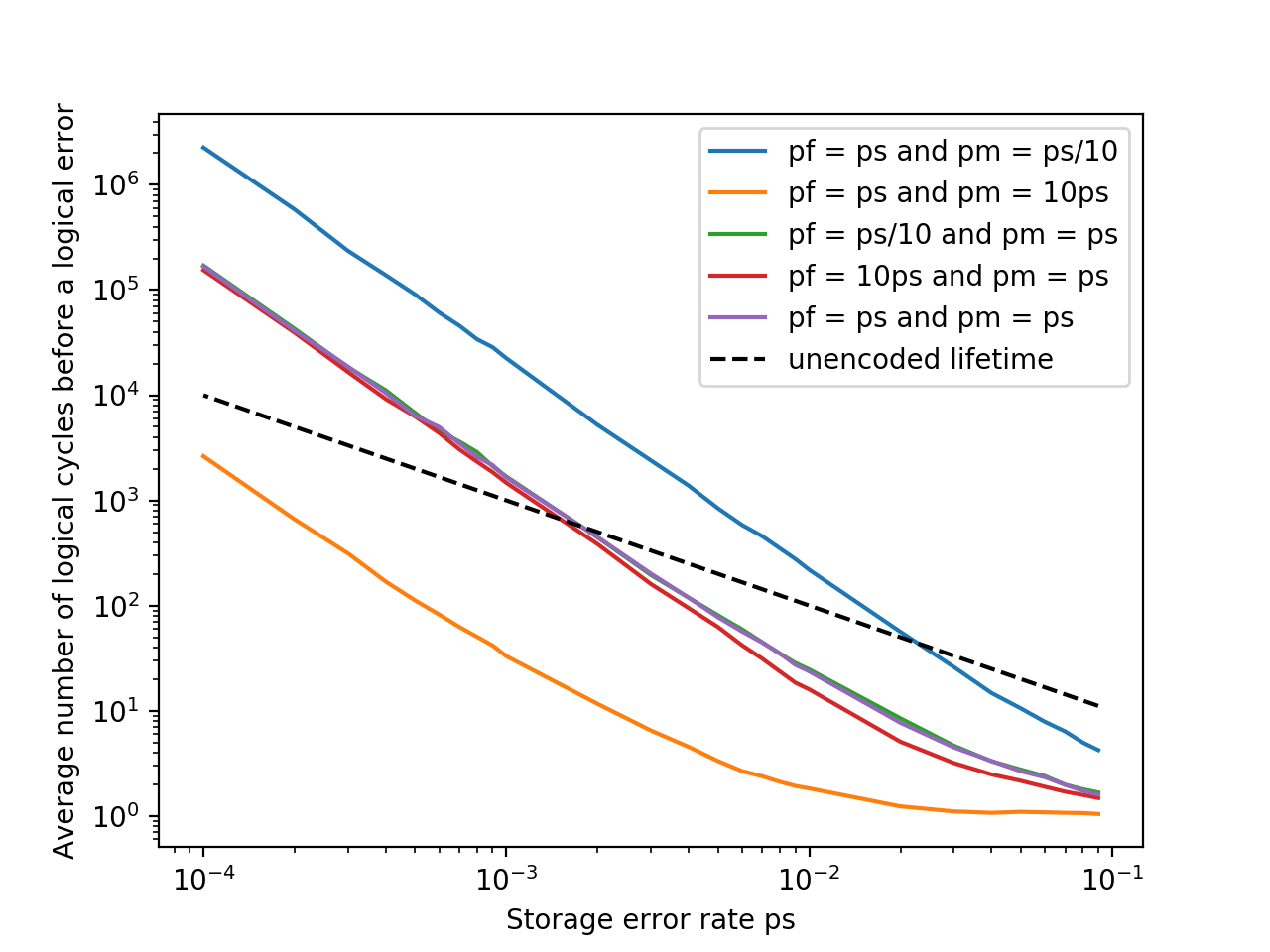}
\caption{Average lifetime of encoded data with the Hamming code combined
with the [5,3,2] measurement code for different noise parameters
obtained over 10,000 trials.
We vary the strength of the measurement noise $p_f$ and the internal noise $p_m$.
Increasing or decreasing the value of  $p_f$ keeps the lifetime roughly unchanged.
However, the encoded lifetime is very sensitive to the internal noise $p_m$.
}
\label{fig:life_time_multivariate}
\end{figure}

\section{Conclusion}

We have introduced a simple formalism for fault-tolerant error correction with 
linear codes. Based on a notion of minimum distance and a decoding 
strategy adapted to this model and we obtain bounds on the time-overhead
for fault tolerance. Our work suggests further extensions.
\begin{itemize}
\item {\em Efficient fault-tolerant decoding algorithm:}
We chose to implement the truncated MWE decoder as via a lookup table which is extremely fast
but the large amount of memory required is an important drawback of this approach,
restricting our scheme to short-length codes. 
It is unclear whether a general efficient implementation of the truncated MWE decoder 
exists.
However, an efficient decoder can be designed for specific families of measurement matrices.
This would significantly extend the scope of the current fault-tolerant error correction scheme.

\item {\em Most Likely Coset decoder:} We observed in Section~\ref{subsec:FTDECODING:circuit_distance}
that some circuit errors are trivial. That means that cosets of circuit errors are indistinguishable.
Identifying the most likely coset instead of the most likely circuit error would lead 
to a better decoder.
In the quantum setting, the equivalence between two errors that differ in a stabilizer should
also be considered. This is another notion of coset that should be exploited in an ideal decoder \cite{dennis2002:tqm, bacon2017:QLDPC_from_circuit, pryadko2019MLD_circuit}.

\item {\em Optimal time overhead:}
Theorem~\ref{theo:time_overhead} provides an upper bound
$O(d_D \log(d_D))$
on the number of measurements required for fault-tolerant error correction
with codes with polynomial distance.
We conjecture that the time overhead can be reduced further with fault-tolerance
sequences with length $O(d_D)$ (Conjecture~\ref{conjecture}).

\item {\em Space-time tradeoff:}
In order to reduce the time overhead, it may be advantageous to encode data 
in a code with large length $n_D$ and suboptimal minimum distance $d_D$.
We leave the study of the tradeoff between the space overhead $n_D / k_D$
and the time overhead $n_M / k_D$ for future research.

\item {\em Beyond sequential measurements:}
We use the number of measurements as a proxy to measure the time-overhead 
for fault-tolerance. This is a good estimate for the number of time-steps required
for a correction round only if most measurement cannot be implemented simultaneously
which is the case for general unstructured codes or code with dense parity-check
matrices.
Using quantum LDPC codes \cite{mackay2004:QLDPC, tillich2013:QLDPC}, 
defined by sparse parity-check matrices, 
a measurement round can be implemented in constant depth while preserving 
a low residual noise after correction \cite{kovalev2013:qldpc, gottesman2014:ldpc, leverrier2015:quantum_expander_codes, fawzi2018:constant_overhead}.
The main reason for this difference is that internal errors only affect $O(1)$ outcome bits.
Consequently, one can ignore internal error and replace them
by low weight correlated measurement errors.
This relates to the notion of single-shot quantum error correction \cite{bombin2015:single_shot, campbell2019:single_shot}
extensively used for topological quantum codes 
\cite{brown2016:FT_guauge_color_codes, breuckmann2016:2D_4D_decoders, kubica2018:CA_decoder, zeng2019:conv_DS_codes}.
In the case of quantum LDPC codes, one may consider minimizing the 
total number of measurements implemented in addition to the measurement 
sequence length.
\end{itemize}

This work is motivated by quantum computing applications, where
quantum error correction must be implemented via hardware 
suffering from high noise rate. 
This work may also find applications in other settings for error correction 
in very noisy environment, for example in flash storage, where the noise rate of 
the densest flash cells reaches $10^{-4}$ or more \cite{cai2017:flash}.

\section*{Acknowledgments}

The authors would like to thank Michael Beverland, Vadym Kliuchnikov, Dave Probert, Alan Geller and David Poulin for their invaluable feedback and Robin Kothari for suggesting an improvement of the bound obtained
in Theorem~\ref{theo:time_overhead}.


\newcommand{\SortNoop}[1]{}

\end{document}